\PassOptionsToPackage{prologue,table,xcdraw}{xcolor}
\PassOptionsToPackage{prologue,table,xcdraw}{color}
\documentclass[sigconf]{acmart}

\usepackage{graphicx}
\usepackage{subcaption}
\usepackage{microtype}
\usepackage{balance}
\usepackage{stfloats}
\usepackage{makecell}
\usepackage{hyperref}
\usepackage{url}
\usepackage{graphicx}
\usepackage{wrapfig}
\usepackage{booktabs}
\usepackage{multirow}
\usepackage{algorithm}
\usepackage{algpseudocode}
\usepackage{amsmath}
\usepackage{amsfonts}
\usepackage{multicol}
\usepackage{enumitem}
\usepackage{pifont}
\usepackage{wrapfig}
\usepackage{xcolor}
\usepackage{soul}
\usepackage[most]{tcolorbox}
\usepackage{colortbl}
\usepackage{lipsum}
\usepackage{mathtools}
\sethlcolor{gray!20}      % 设置高亮底色

\newcommand{\PreserveBackslash}[1]{\let\temp=\\#1\let\\=\temp}
\newcolumntype{C}[1]{>{\PreserveBackslash\centering}p{#1}}
\newcolumntype{R}[1]{>{\PreserveBackslash\raggedleft}p{#1}}
\newcolumntype{L}[1]{>{\PreserveBackslash\raggedright}p{#1}}

\newtheorem{theorem}{Theorem}

\newcommand{\ourmethod}{APAO}

\AtBeginDocument{%
  }

\copyrightyear{2026}
\acmYear{2026}
\setcopyright{cc}
\setcctype{by}
\acmConference[KDD '26]{Proceedings of the 32nd ACM SIGKDD Conference on Knowledge Discovery and Data Mining V.2}{August 09--13, 2026}{Jeju Island, Republic of Korea}
\acmBooktitle{Proceedings of the 32nd ACM SIGKDD Conference on Knowledge Discovery and Data Mining V.2 (KDD '26), August 09--13, 2026, Jeju Island, Republic of Korea}
\acmISBN{979-8-4007-2259-2/2026/08}
\acmDOI{10.1145/3770855.3818052}

\begin{document}

\title{APAO: Bridging the Training-Inference Gap in Generative Recommendation via Adaptive Prefix-Aware Optimization}

\author{Yuanqing Yu}
\affiliation{
  \institution{DCST, Tsinghua University}
  \city{Beijing 100084}
  \country{China}
}
\affiliation{
  \institution{Quancheng Laboratory}
  \city{Jinan, Shandong}
  \country{China}
}
\email{yyq23@mails.tsinghua.edu.cn}

\author{Yifan Wang}
\affiliation{
  \institution{DCST, Tsinghua University}
  \city{Beijing 100084}
  \country{China}
}
\email{yf-wang21@mails.tsinghua.edu.cn}

\author{Weizhi Ma}
\authornote{Corresponding author.
% This work is supported by the Natural Science Foundation of China (Grant No. 62372260) and the National Key Research and Development Program of China under Grant 2024YFC3307403. Weizhi Ma is also sponsored by the Beijing Nova Program.
}

\affiliation{%
  \institution{AIR, Tsinghua University}
  \city{Beijing 100084}
  \country{China}
}
\email{mawz@tsinghua.edu.cn}

\author{Zhiqiang Guo}
\affiliation{
  \institution{DCST, Tsinghua University}
  \city{Beijing 100084}
  \country{China}
}
\email{georgeguo.gzq.cn@gmail.com}

\author{Min Zhang}
\authornotemark[1]
\affiliation{
  \institution{Quancheng Laboratory}
  \city{Jinan, Shandong}
  \country{China}
}
\affiliation{%
  \institution{DCST, Tsinghua University}
  \city{Beijing 100084}
  \country{China}
}
\email{z-m@tsinghua.edu.cn}

\renewcommand{\shortauthors}{Yuanqing Yu, Yifan Wang, Weizhi Ma, Zhiqiang Guo, \& Min Zhang}




\begin{abstract}
Generative recommendation has recently emerged as a promising paradigm for sequential recommendation. It formulates the task as an autoregressive generation process, predicting tokens of the next item conditioned on user interaction histories.
% 问题
Existing generative recommendation models are typically trained with token-level likelihood objectives such as cross-entropy loss, while employing beam search during inference to generate ranked candidates.
However, this leads to a fundamental \textbf{training-inference inconsistency}: standard training assumes ground-truth tokens are always available, while beam search prunes low-probability branches during inference, causing the correct item to be prematurely discarded when its prefixes receive low scores.
% Ours
To address this issue, we propose the \textbf{Adaptive Prefix-Aware Optimization (\ourmethod{})} framework, which introduces prefix-level optimization losses to better align the training objective with the inference setting.
Furthermore, we design an adaptive worst-prefix optimization strategy that dynamically focuses on the most vulnerable prefixes during training, thereby enhancing the model's ability to retain correct candidates under beam search constraints.
We provide theoretical analyses to demonstrate the effectiveness and efficiency of our framework. Extensive experiments show that \ourmethod{} consistently alleviates the training-inference inconsistency and improves performance across generative recommendation backbones.
The source code is publicly available at \href{https://github.com/yuyq18/APAO}{https://github.com/yuyq18/APAO}.
\end{abstract}


\maketitle
\newcommand\kddavailabilityurl{https://doi.org/10.5281/zenodo.20439389}
\ifdefempty{\kddavailabilityurl}{}{
\begingroup\small\noindent\raggedright\textbf{Resource Availability:}\\
% please change the following context to include multiple artifacts if necessary, including data, models, code, etc.
The source code of this paper has been made publicly available at \url{\kddavailabilityurl} and \url{https://github.com/yuyq18/APAO}.
\endgroup
}

% \section*{Relevance}
% % Relevance
% This work improves generative recommendation on the web through prefix-aware optimization, aligning with the User Modeling, Personalization, and Recommendation track.

\begin{figure}[t]
    \centering
    \includegraphics[width=1.0\linewidth]{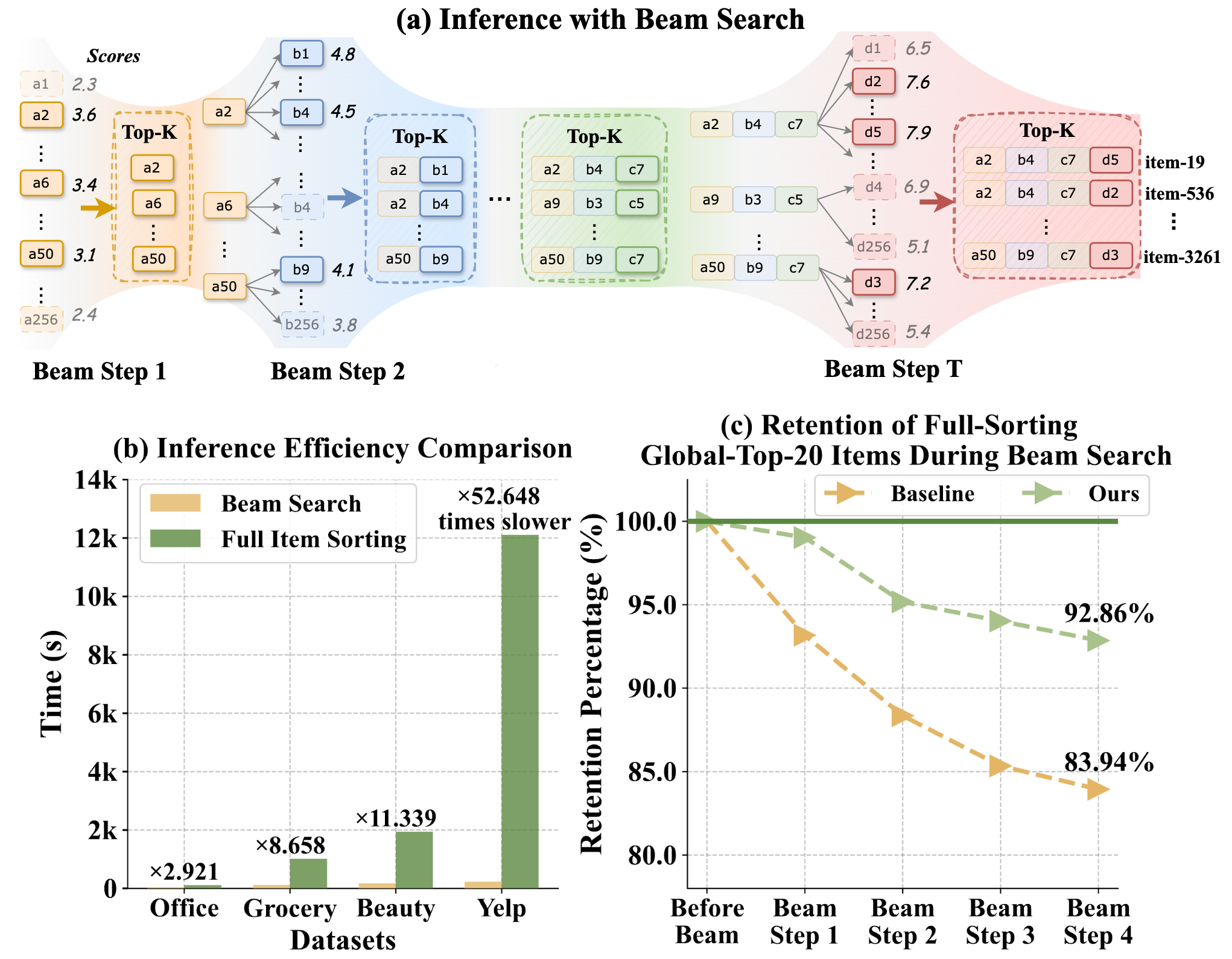}
    \caption{Analysis of Beam Search. (a) Illustration of the beam search inference process. (b) Efficiency comparison: Beam Search vs. Full Sorting. (c) Retention of global Top-20 items: Our method (green) effectively prevents the early discarding of optimal items observed in the baseline (yellow).}
    \label{fig:observation}
\end{figure}

\section{Introduction}

% 1.
In recommender systems (RSs), modeling the sequential patterns of user behavior is crucial for accurately predicting future interactions of users. This problem setting is commonly referred to as sequential recommendation~\cite{kang2018sasrec, hidasi2015gru4rec, sun2019bert4rec, tang2018Caser}. 
Inspired by the recent success of generative language models~\cite{yang2023AutoSearchIndexer, sun2023GenRetrie, zeng2022glm, gpt42023technical}, an emerging line of research has explored the generative paradigm in recommender systems~\cite{rajput2023tiger, wang2024letter, liu2025ETEGRec, zhou2025onerecv2}.
Unlike traditional discriminative methods, generative approaches tokenize each item into multiple discrete tokens and train models to autoregressively generate the target token sequence conditioned on the user’s interaction history.

% 2.
Beam search serves as the mainstream decoding strategy in generative recommendation~\cite{rajput2023tiger, wang2024letter, renzhaochun2025constrained, liu2025ETEGRec, zhou2025onerecv2}, driven by the dual advantages of \textbf{relevance-oriented Top-$K$ ranking} and \textbf{inference efficiency}.
Unlike the stochastic sampling-based decoding strategies widely used in large language models (LLMs) to enhance diversity~\cite{gpt42023technical, dubey2024llama}, generative recommendation prioritizes the likelihood of relevant items to ensure accurate retrieval, which naturally aligns with beam search (as illustrated in \autoref{fig:observation}(a)).
Moreover, directly computing likelihoods over the entire item corpus (\textit{Full Item Sorting}) incurs prohibitive computational costs as the corpus scales.
Beam search serves as an efficient approximation by retaining only the Top-$K$ high-probability prefixes at each decoding step (as illustrated in \autoref{fig:observation}(a)).
As shown in \autoref{fig:observation}(b), beam search achieves orders-of-magnitude faster inference than \textit{Full Item Sorting}.

% 3. Motivation
Although beam search offers substantial efficiency gains, it introduces a fundamental \textbf{training-inference inconsistency}, arising from the \textbf{teacher-forcing} training paradigm and the \textbf{pruning-based} inference process.
Specifically, standard training assumes that the ground-truth prefix is always available, while beam search aggressively prunes low-scoring prefixes during inference.
As a result, the model is not explicitly optimized to retain globally relevant items under beam constraints: an item with high overall likelihood may still be prematurely discarded because its early-stage prefixes fail to rank in the Top-$K$.
As empirically verified in \autoref{fig:observation}(c), some items ranked in the global Top-20 by \textit{Full-Item Sorting} are eliminated during intermediate beam search steps.

% Ours
To address this issue, we conduct a formal analysis comparing beam search with full-space ranking, showing that beam search introduces additional prefix-level ranking constraints, which are absent from standard training objectives.
To mitigate this problem, we propose the \textbf{Adaptive Prefix-Aware Optimization (\ourmethod{})} framework, which introduces prefix-level optimization losses during training to better align the model with its inference behavior.
While inference-side solutions may alleviate this issue, they often introduce additional serving latency. Therefore, instead of modifying the beam-search pipeline, we focus on a training-side solution that aligns optimization with beam-search inference without increasing inference cost.
In developing this framework, we resolve two key challenges: how to design effective prefix-aware objectives, and how to adaptively prioritize different prefixes during optimization.
For the first challenge, motivated by the ranking nature of beam search, we design prefix-level objectives that explicitly model the progressive pruning behavior of beam search. Specifically, the model is encouraged to maintain high-quality prefixes through pointwise and pairwise ranking optimization.
For the second challenge, we observe that the effectiveness of beam search largely depends on whether suboptimal prefixes can still preserve the correct candidates during the decoding process.
Accordingly, we propose adaptive worst-prefix optimization, which focuses the optimization on the currently weakest prefix during training.

Furthermore, we provide a theoretical analysis, including a lower bound derivation and a discussion on time complexity. This analysis demonstrates that the proposed framework effectively optimizes evaluation metrics under the beam-search strategy, while maintaining a complexity comparable to existing baselines.
% Experiments
Extensive experiments on public datasets with two representative backbones demonstrate that our method consistently improves recommendation performance under beam-search decoding. In-depth analysis further reveals consistent gains at every decoding step (prefix), effectively mitigating the training–inference inconsistency.
Moreover, online A/B testing shows that APAO-Pointwise improves pCTR by 0.9\% in a real-world industrial recommendation system.

The main contributions of this work are summarized as follows:

\begin{itemize}[leftmargin=*]
    % Observation. 
    \item To the best of our knowledge, this work presents the first investigation of the \textbf{training-inference inconsistency} in generative recommendation. We identify the prefix-level misalignment inherent in existing approaches and provide a principled solution.
    % Method.
    \item We propose \textbf{Adaptive Prefix-Aware Optimization} (\ourmethod{}), an adaptive optimization framework that incorporates prefix-level objectives to better account for beam-search inference.
    We further provide a theoretical analysis demonstrating its effectiveness and computational efficiency.
    % Results.
    \item Extensive experiments on diverse backbones demonstrate the effectiveness of \ourmethod{} in mitigating the training-inference inconsistency. Moreover, online deployment achieves +0.9\% pCTR in a real-world industrial system.
\end{itemize}

\section{Related Work}
\subsection{Generative Recommendation}
Generative recommendation (GR) has emerged as a new paradigm for recommender systems, where items are represented by multiple discrete tokens instead of a single ID, improving parameter efficiency and generalization.
Existing GR studies mainly involve two stages: item tokenization and sequence modeling over discrete tokens.
For tokenization, prior methods include rule-based approaches~\cite{hua2023index} and learning-based approaches~\cite{rajput2023tiger, wang2024letter, liu2025ETEGRec}, among which VQ-based methods (e.g., RQ-VAE~\cite{lee2022rqvae}) are dominant.
Recent studies further improve tokenization through richer input signals~\cite{liu2024multi, zhai2025multimodal}, recommendation-aware objectives~\cite{liu2025ETEGRec}, and context-adaptive representations~\cite{houactionpiece}.
For sequence modeling, most methods adopt autoregressive next-token prediction conditioned on user interaction histories.
Recent efforts mainly focus on improving inference efficiency, stability, and scalability through enhanced decoding strategies and model designs~\cite{lin-efficient, lin2025order, yang2025earn}.

Unlike prior studies, our work focuses on the training--inference inconsistency introduced by beam search decoding, which has not been systematically explored in existing GR research.

% Optimization
\subsection{Alignment and Ranking Optimization in Generative Models}

In LLMs, preference alignment methods such as Reinforcement Learning from Human Feedback (RLHF)~\cite{christiano2017deep, ouyang2022training} align model outputs with human preferences.
To improve stability and efficiency, Direct Preference Optimization (DPO)~\cite{rafailov2023DPO} and its extensions~\cite{azar2024general, ethayarajh2024kto} replace reinforcement learning with closed-form optimization.
% genrec
Several studies~\cite{chen2024SDPO, bai2024DMPO, wang2025msl} refine GR training objectives to better align optimization with recommendation goals.
Among them, preference-based methods such as S-DPO~\cite{chen2024SDPO} and DMPO~\cite{bai2024DMPO} introduce ranking optimization stages to capture relative order information \textbf{at the whole-item or complete-text level}.
Furthermore, MSL~\cite{wang2025msl} modifies the cross-entropy loss by masking invalid tokens in the softmax normalization to reduce noise from uninformative tokens.

Different from these studies that focus on coarse-grained item-level optimization, our work identifies a critical gap: the \textbf{prefix-level} training–inference inconsistency.
Due to the sequential nature of autoregressive decoding, relying solely on full-sequence supervision neglects errors accumulated at intermediate decoding steps.
We address this by introducing training objectives aligned with the inference process, which explicitly improve ranking capabilities across all decoding steps, from partial prefixes to complete sequences.

\section{Preliminaries}
In this section, we first formulate the training objective and inference process of Generative Recommendation (GR). 
Subsequently, we analyze the discrepancy between beam search decoding and full-space ranking to elucidate the fundamental training–inference inconsistency.

\subsection{Training Objective of GR}

In the sequential recommendation scenario, a user’s interaction history is denoted as $x = \{y_1, y_2, \ldots, y_{i-1}\}$, representing a sequence of past interactions.
In generative recommendation, the target item $y_i$ is tokenized into a sequence of $T$ discrete tokens, $y_i = \{y_i^1, y_i^2, \ldots, y_i^T\}$, using rule-based or vector quantization methods (e.g., RQ-VAE~\cite{rajput2023tiger}).
We denote $y_i^t$ as the $t$-th token and $y_i^{<t} = \{y_i^1, \ldots, y_i^{t-1}\}$ as its prefix subsequence.
The learning objective of GR is to estimate the conditional probability distribution of the next token given both the user history and the previously generated prefix tokens:
\begin{equation}
P(y_i^t \mid y_i^{<t}, x).
\label{equ:conditional_prob}
\end{equation}

By autoregressively generating tokens until step $T$, the model constructs the complete item $y_i = \{y_i^{1:T}\}$.

Analogous to the training of large language models~\cite{ouyang2022training, gpt42023technical, dubey2024llama}, GR models are typically optimized using the \textbf{cross-entropy (CE) loss} under a \textbf{teacher-forcing} paradigm.
% Specifically, at each step $t$, the model $f_\theta$ predicts the logits $z_t$ for the next token, conditioned on the user history $x$ and the \textbf{ground-truth prefix} $y_i^{<t}$:
% \begin{equation}
% z_i^t = f_\theta(y_i^t \mid y_i^{<t}, x),
% \end{equation}
Specifically, at each step $t$, the model $f_\theta$ predicts a logit distribution over the token vocabulary conditioned on the user history $x$ and the \textbf{ground-truth prefix} $y_i^{<t}$.
Let $z_{i,v}^t$ denote the logit assigned to token $v \in V$:
\begin{equation}
z_{i,v}^t = f_\theta(v \mid y_i^{<t}, x).
\end{equation}
where $f_\theta$ is parameterized by $\theta$. The probability assigned to the ground-truth token $y_i^t$ is computed via the softmax function over the entire token vocabulary $V$:
\begin{equation}\label{equ:s_i^t}
s_i^t = \log \frac{e^{z_{i,y_i^t}^t}}{\sum\limits_{j \in V} e^{z_{i,j}^t}}.
\end{equation}

Given a target item represented by a sequence of $T$ tokens, the CE loss is defined as
\begin{equation}\label{equ:ce_loss}
% \mathcal{L}_{CE} = -\frac{1}{T}\sum_{t=1}^{T} s_i^t 
% = -\frac{1}{T}\sum_{t=1}^{T} \log \frac{e^{z_i^t}}{\sum\limits_{j \in V} e^{z_j^t}}.
\mathcal{L}_{CE} =
-\frac{1}{T}\sum_{t=1}^{T}
\log
\frac{
e^{z_{i,y_i^t}^t}
}{
\sum\limits_{j \in V} e^{z_{i,j}^t}
}.
\end{equation}
This objective maximizes the likelihood of the ground-truth tokens while suppressing competing tokens in the vocabulary.

\subsection{Inference Process of GR}
% beam search
During inference, beam search~\cite{renzhaochun2025constrained} is widely adopted in GR for its efficiency, emphasizing high-likelihood relevant items while generating multiple candidate sequences in a single decoding process.

Given a beam size of $K$, at each decoding step $t$, the model maintains up to $K$ candidate sequences (beams) denoted as $\mathcal{B}_{t-1} = \{y^{<t,(1)}, y^{<t,(2)}, \ldots, y^{<t,(K)}\}$ associated with their cumulative scores
$\{S^{(1)}_{t-1}, S^{(2)}_{t-1}, \ldots, S^{(K)}_{t-1}\}$.

For each beam $k$, the probability of extending it with a candidate token $v \in V$ is given by:
\begin{equation}
P(v \mid y^{<t,(k)}, x) = \frac{e^{f_\theta(v \mid y^{<t,(k)}, x)}}{\sum\limits_{j \in V} e^{f_\theta(j \mid y^{<t,(k)}, x)}}.
\end{equation}

The extended hypothesis $\{y^{<t,(k)} \oplus v\}$ is assigned a new cumulative score:
\begin{equation}
S^{(k)}_t(v) = S^{(k)}_{t-1} + \log P(v \mid y^{<t,(k)}, x).
\end{equation}

Across all beams and candidate tokens, the algorithm selects \textbf{the top-$K$ extensions} with the highest cumulative scores:
\begin{equation}\label{eq:beam_topk}
\mathcal{B}_t = \{ y^{<t,(k)} \oplus v: (k, v) \in 
\operatorname{TopK}\!\big( \{ S^{(k)}_t(v) \mid 1 \leq k \leq K, v \in V \} \big)\}.
\end{equation}

This process is repeated until $T$ tokens are generated, resulting in $K$ complete item hypotheses $\{y^{(1)}, y^{(2)}, \ldots, y^{(K)}\}$ with scores $\{S_T^{(1)}, S_T^{(2)}, \ldots, S_T^{(K)}\}$.  
For recommendation, the top-$K$ items ranked by their cumulative scores are returned as the final results.

\subsection{Training–Inference Inconsistency Analysis}
\label{sec: analysis_inconsistency}

In this subsection, we formally analyze the training–inference inconsistency. Taking \textbf{Recall@K} as a representative metric, we contrast the distinct conditions required for a target item $y$ to be successfully recalled in the ideal full-space ranking versus the practical beam search decoding.

\vspace{0.1cm}
\noindent \textbf{Full-Space Ranking (Global View).}
Ideally, if we compute scores for all items (\textit{Full-Sorting}), the target item $y$ is successfully recalled if its total score ranks within the top-$K$ of the entire item space $\mathcal{Y}$.
Let $S(y|x) = \sum_{t=1}^T \log P(y^t | y^{<t}, x)$ be the total score derived from Eq.~(\ref{equ:ce_loss}). The success condition is:
\begin{equation}\label{eq:full_recall}
    \mathbb{I}_{\text{Full}}(y) = \mathbb{I}\left( \operatorname{rank}_{\mathcal{Y}}(S(y|x)) \le K \right).
\end{equation}
Note that Eq.~(\ref{eq:full_recall}) depends only on the \textbf{final cumulative summation}. Even if an intermediate prefix has a low probability, the item can still be recalled if subsequent tokens compensate for it.

\vspace{0.1cm}
\noindent \textbf{Beam Search (Local Constraints).}
In contrast, beam search imposes a strict \textbf{step-wise ranking constraint}.
At each decoding step $t$, the prefix $y^{1:t}$ must compete with other candidates. Let $\operatorname{rank}_t(y^{1:t})$ denote the rank of the prefix $y^{1:t}$ among the current candidates based on the cumulative score up to step $t$.
For the final item $y$ to be generated, its partial sequence must rank within the top-$K$ at every single step $t$.
The success condition thus becomes a chain of intersection constraints:
\begin{equation}\label{eq:beam_recall}
    \mathbb{I}_{\text{Beam}}(y) = \prod_{t=1}^T \mathbb{I}\left( \operatorname{rank}_{t}(S_t(y^{1:t})) \le K \right).
\end{equation}
Here, $S_t(y^{1:t}) = \sum_{j=1}^t \log P(y^j | y^{<j}, x)$ is the partial cumulative score.
Unlike Eq.~(\ref{eq:full_recall}), Eq.~(\ref{eq:beam_recall}) implies that a single ranking failure at any intermediate step (i.e., $\operatorname{rank}_t > K$) results in the permanent loss of the item \textbf{($\mathbb{I}_{\text{Beam}} = 0$)}, regardless of its potential final score.

\vspace{0.1cm}
\noindent \textbf{The Inconsistency.}
Comparing Eq.~(\ref{eq:full_recall}) and Eq.~(\ref{eq:beam_recall}) reveals the fundamental conflict:
The CE loss (Eq.~\ref{equ:ce_loss}) optimizes the \textbf{average} token likelihood (allowing trade-offs), whereas beam search requires satisfying \textbf{all step-wise ranking constraints} and is intolerant to failures at intermediate steps.
Since the model is trained with ground-truth prefixes (teacher-forcing) and never penalized for "falling out of the beam," it fails to learn robustness against beam pruning, leading to the premature discarding of relevant items.
\section{APAO: Adaptive Prefix-Aware Optimization}

In this section, we present our \textbf{Adaptive Prefix-Aware Optimization (\ourmethod{})} framework, addressing the training–inference inconsistency identified in~\autoref{sec: analysis_inconsistency}.
% 4.1
We first provide an overview of the \ourmethod{} framework, which explicitly integrates prefix-level optimization objectives to align the training process with beam search inference.
%4.2 4.3
Subsequently, we detail our technical solutions to two core challenges: constructing an effective prefix-aware loss function (\autoref{method: prefix-aware loss}) and adaptively determining the relative importance of different prefixes (\autoref{method: adaptive weight}).
% 4.4
Finally, we provide theoretical analyses of \ourmethod{}, including a lower bound derivation and time complexity discussion (\autoref{sec: theoretical_analysis}).

\subsection{Prefix-Aware Optimization}
To mitigate the aforementioned training-inference inconsistency, it is essential to incorporate prefix-level optimization objectives into the training process. 
The core idea is to encourage the model to maintain high probabilities not only for the final token sequence but also for all valid prefixes that are required to survive during beam decoding.

Formally, we define a prefix-aware optimization objective that aggregates losses from different prefix lengths:
\begin{equation}
    \mathcal{L}_{\text{prefix}} \;=\; \sum_{m=1}^{T} w_m \, \mathcal{L}_{m},
\end{equation}

where $\mathcal{L}_{m}$ denotes the loss corresponding to the $m$-th prefix and $w_m \ge 0$ controls the contribution of each prefix length $m$ to the final loss.

We incorporate prefix-level optimization via a unified one-stage training framework. Specifically, we formulate the final objective as a weighted combination of the original CE loss and the proposed prefix-aware loss:
\begin{equation}\label{equ:overall_loss}
    \mathcal{L}_\text{unified}
=
\mathcal{L}_{\mathrm{CE}}
+
\beta \sum_{m=1}^{T} w_m \mathcal{L}_{m},
\end{equation}
where $\beta \ge 0$ is a hyperparameter that controls the trade-off between the token-level likelihood (CE) and the prefix-level ranking constraints.

With this prefix-aware training paradigm established, the remaining questions are what specific form of prefix-aware loss $\mathcal{L}_{m}$ to use (Section~\ref{method: prefix-aware loss}) and how to determine the training weight $w_m$ for each prefix (Section~\ref{method: adaptive weight}).

\subsection{Prefix-Aware Loss: Pointwise and Pairwise}
\label{method: prefix-aware loss}

Fundamentally, prefix optimization can be formulated as a ranking problem. While ranking objectives are typically categorized into pointwise, pairwise, and listwise approaches, listwise methods generally incur substantial computational costs. Consequently, to ensure training efficiency, this work focuses on pointwise and pairwise prefix-aware losses.

\vspace{0.1cm}
\noindent \textbf{Prefix-aware Pointwise Loss.} Analogous to the standard CE loss defined over the full sequence (Eq.~\ref{equ:ce_loss}), we define the pointwise loss for a specific prefix of length $m$ as:
\begin{equation}\label{equ:point_loss_def}
\begin{aligned}
\mathcal{L}_\mathrm{point}{(m)}
=
-\frac{1}{m}\sum_{t=1}^{m}
\log
\frac{
e^{z_{i,y_i^t}^t}
}{
\sum\limits_{j \in V} e^{z_{i,j}^t}
},
\quad m \in \{1, \dots, T\},
\end{aligned}
\end{equation}
where $z_{i,j}^t$ denotes the logit assigned to token $j \in V$ at step $t$.

The pointwise prefix-aware loss can be interpreted as assigning different optimization emphasis across decoding positions, thereby encouraging stronger prefix-level robustness during beam search.

\vspace{0.1cm}
\noindent \textbf{Prefix-aware Pairwise Loss.} 
To explicitly distinguish the ground-truth path from incorrect paths during decoding, we introduce a pairwise ranking objective, optimizing the \textbf{relative ranking} between the positive prefix and negative prefixes at each step.

For a target positive item $y_+$, we consider its prefix of length $m$, denoted as $y_+^{1:m}$.
To construct the corresponding negative candidates, we sample a set of negative items $\mathcal{N}$ from the corpus and truncate them to the same length $m$.
This yields a set of negative prefixes $\{ y^{1:m}_{j,-} \mid j \in \mathcal{N} \}$, enabling prefix-level comparison under the same decoding depth.

We define the cumulative score for a prefix of length $m$ as the sum of token-level log-probabilities:
\begin{equation}
    S_i^m = \sum_{t=1}^m s_i^t,
\end{equation}
where $s_i^t$ is the token-level score given in Eq.~\eqref{equ:s_i^t}.

Based on the prefix score, the prefix-aware pairwise loss compares the positive item 
against negatives under the same prefix length:
\begin{equation}\label{equ:pair_loss_def}
    \mathcal{L}_{\mathrm{pair}}(m) 
    = -\log \sigma \left(- \log \sum_{j \in \mathcal{N}} \exp\left(S_{j,-}^m - S_{i,+}^m \right)\right),
    \quad m \in \{1, \dots, T\},
\end{equation}
where $S_{i,+}^m$ denotes the cumulative score of the positive prefix, and $S_{j,-}^m$ denotes the cumulative score of the $j$-th negative prefix.

It is worth noting that while our pairwise loss appears structurally similar to S-DPO~\cite{chen2024SDPO}, they differ fundamentally in their optimization granularity, as detailed in~\autoref{appendix: discussion}.

\vspace{0.1cm}
\noindent \textbf{Comparison}. The pointwise loss is more efficient since it does not require negative sampling for positive prefixes; however, its drawback is that it focuses more on data distribution cloning rather than modeling relative ranking relationships, which are crucial for prefix optimization.
In contrast, the pairwise loss better captures relative ranking relationships, but comes at the cost of lower computational efficiency.
We conduct detailed analyses of both methods in the experimental section.

\subsection{Adaptive Worst-prefix Optimization}
\label{method: adaptive weight}

The remaining challenge lies in how to determine the appropriate weights $\{w_m\}_{m=1}^T$ for prefixes of different lengths.
A naive solution is to treat $w$ as a set of tunable hyperparameters; however, this approach becomes impractical due to the large number of prefixes and the resulting combinatorial tuning cost.
To address this limitation, we propose an \textbf{adaptive weighting mechanism} that dynamically adjusts the weights of each prefix during training.

Notably, Eq.~(\ref{eq:beam_recall}) shows that if any single prefix fails to survive beam search, the entire decoding path is eliminated. Therefore, the overall decoding process is effectively bottlenecked by the weakest prefix, and the optimization should emphasize the prefix with the currently worst performance.
A direct formulation is to minimize the worst-case loss (Hard-Max):
\begin{equation}
    \mathcal{L}_{\mathrm{worst}} \;=\; \max_{m \in \{1, \dots, T\}} \mathcal{L}_{m},
\end{equation}
where $\mathcal{L}_m$ can correspond to either the pointwise or pairwise prefix-aware loss.

However, directly identifying the worst prefix based on instantaneous mini-batch losses may lead to instability.
To mitigate this issue, inspired by~\cite{piratla2022focus}, we introduce a soft weighting vector $\mathbf{w} \in \triangle_T$ that smooths the prefix selection process and constrains its variation across consecutive steps.
The weights are updated by solving the following optimization problem:
\begin{equation}\label{eq:optimization_problem}
    \mathbf{w}^{(\tau+1)}
=
\arg \max_{\mathbf{w} \in \triangle_T}
\sum_m w_m \mathcal{L}^{(\tau+1)}_m
-
\frac{1}{\eta}
KL(\mathbf{w}\,\|\,\mathbf{w}^{(\tau)}),
\end{equation}
where $KL(\mathbf{w}\,\|\,\mathbf{w}^{(\tau)})$ represents the Kullback-Leibler divergence between the current and previous weight vectors, and $\triangle_T$ denotes the T-dimensional probability simplex.
The hyperparameter $\eta$ controls the degree of variation in $\mathbf{w}$; smaller values of $\eta$ yield smoother and more stable updates.

The above problem admits a closed-form solution derived from the Karush–Kuhn–Tucker (KKT) optimality conditions. The detailed proof can be found in~\autoref{appendix: kkt}:
\begin{equation}\label{eq:weight_update}
w^{(\tau+1)}_m
=
\frac{
w^{(\tau)}_m \cdot \exp(\eta \mathcal{L}_m)
}{
\sum_{j=1}^T
w^{(\tau)}_j \cdot \exp(\eta \mathcal{L}_j)
}.
\end{equation}

Intuitively, the adaptive weighting scheme assigns larger weights to prefixes with higher losses, allowing the model to dynamically focus on the most difficult prefixes during training.

By integrating the adaptive weighting mechanism with the prefix-aware loss functions in~Eq.~\ref{equ:overall_loss}, we derive the unified training objectives for the two variants of \ourmethod{}.

\noindent 1. Pointwise Variant:
\begin{equation}\label{equ:final_pointwise}
    \mathcal{L}_{\text{total}} = \mathcal{L}_{\text{CE}} + \beta \sum_{m=1}^{T} w_m \mathcal{L}_{\text{point}}(m),
\end{equation}
where $\mathcal{L}_{\text{point}}(m)$ is the pointwise loss defined in Eq.~(\ref{equ:point_loss_def}).

\noindent 2. Pairwise Variant:
\begin{equation}\label{equ:final_pairwise}
    \mathcal{L}_{\text{total}} = \mathcal{L}_{\text{CE}} + \beta \sum_{m=1}^{T} w_m \mathcal{L}_{\text{pair}}(m),
\end{equation}
where $\mathcal{L}_{\text{pair}}(m)$ is the pairwise ranking loss defined in Eq.~(\ref{equ:pair_loss_def}).

In both cases, $\beta \ge 0$ regulates the strength of prefix-level optimization, ensuring that the model learns to satisfy beam search constraints while maintaining generation quality. The weights $\{w_m\}$ are updated iteratively via Eq.~(\ref{eq:weight_update}).
The pseudocode of the complete learning algorithm is provided in~\textbf{\autoref{appendix:learning_algo}}.

\subsection{Theoretical Analysis}
\label{sec: theoretical_analysis}
In this section, we analyze our method from two perspectives: effectiveness and efficiency. We demonstrate that: (i) optimizing the prefix-aware loss effectively optimizes a lower bound of the ranking metric in the beam search setting; and (ii) incorporating prefix-level alignment does not increase the theoretical time complexity compared with other common baseline methods.

\subsubsection{\textbf{Lower-Bound Analysis.}}

\begin{theorem}[Optimization Consistency]\label{thm:consistency}
Optimizing the prefix-aware loss (Eq.~(\ref{equ:overall_loss})) essentially optimizes a lower bound of the ranking metric (e.g., recall) under beam search. It thereby serves as a principled surrogate objective for ranking optimization under beam search.
\end{theorem}

The proof of this theorem is provided in~\textbf{~\autoref{appendix_theorem1}}.

\subsubsection{\textbf{Time Complexity.}}
\label{sec: Time_Complexity_Analysis}

Regarding the training time complexity, it is straightforward to see that the pointwise variant of our method shares the same per-batch complexity as CE loss, both being $O\left(B\cdot T\cdot \left(d^2+ \left| V\right| \cdot d\right)\right)$, where $B$ denotes the batch size, $V$ denotes the token vocabulary, $T$ is the max sequence length, and $d$ is the hidden size. This is because the prefix-level losses introduced by our method can be computed directly based on the token-level logits, without requiring any additional forward passes. 
For the pairwise version with $N$ negatives, the per-batch training complexity is $O\left(B\cdot \left(N+1\right) T\cdot \left(d^2+ \left| V\right| \cdot d\right)\right)$, which is also consistent with other preference-alignment methods such as S-DPO~\cite{chen2024SDPO}.
In addition, we empirically compare the running time of the methods and found that our approach achieves comparable efficiency to the corresponding baselines, shown in~\autoref{sec: efficiency}.

\section{Experiments}

\begin{table*}[!t]
\caption{
Top-$K$ recommendation results on two generative backbones across four datasets.
\textbf{Bold} and \underline{underlined} numbers indicate the best and second-best results, respectively, while \protect\hl{gray-shaded} cells mark the strongest baselines. "R" and "N" denote Recall and NDCG.
The superscripts * indicate $ p \leq 0.01$ for one-sample t-tests of \ourmethod{} vs. the best baseline (the relative improvements are denoted as \textit{Max Improv.}).
}
  \centering
  \renewcommand\arraystretch{1.1}
  \tabcolsep=0.08cm
  \resizebox{\textwidth}{!}{
  % \begin{tabular}{l|cccc|cccc|cccc|cccc}
  \begin{tabular}{lcccccccccccccccc}
    \toprule
    \multirow{2}{*}{\textbf{Methods}} & \multicolumn{4}{c}{\textbf{Office}} & \multicolumn{4}{c}{\textbf{Grocery}} & \multicolumn{4}{c}{\textbf{Beauty}} & \multicolumn{4}{c}{\textbf{Yelp}}\\
    \cmidrule(r){2-5}\cmidrule(r){6-9}\cmidrule(r){10-13}\cmidrule(r){14-17}
    & \textbf{R@10} & \textbf{R@20} & \textbf{N@10} & \textbf{N@20} & \textbf{R@10} & \textbf{R@20} & \textbf{N@10} & \textbf{N@20} & \textbf{R@10} & \textbf{R@20} & \textbf{N@10} & \textbf{N@20} & \textbf{R@10} & \textbf{R@20} & \textbf{N@10} & \textbf{N@20} \\
    \midrule
    \rowcolor[HTML]{e2edf8}
    \multicolumn{17}{c}{\textbf{\textit{TIGER (Encoder-Decoder)}}}\\
    CE & 0.0608 & 0.1013 & 0.0292 & 0.0394 & \cellcolor{gray!20}{0.0775} & \cellcolor{gray!20}{0.1182} & \cellcolor{gray!20}{0.0403} & \cellcolor{gray!20}{0.0506} & \cellcolor{gray!20}0.0611 & 0.0928 & \cellcolor{gray!20}0.0318 & \cellcolor{gray!20}0.0398 & 0.0384 & 0.0585 & 0.0201 & 0.0252 \\
    MSL & \cellcolor{gray!20}0.0638 & 0.0975 & 0.0319 & 0.0404 & 0.0762 & 0.1129 & 0.0401 & 0.0493 & 0.0578 & 0.0873 & 0.0289 & 0.0364 & 0.0384 & \cellcolor{gray!20}0.0586 & \cellcolor{gray!20}0.0206 & \cellcolor{gray!20}{0.0257} \\
    CE$\rightarrow$DPO & 0.0565 & 0.0922 & 0.0275 & 0.0365 & 0.0629 & 0.0983 & 0.0322 & 0.0410 & 0.0503 & 0.0767 & 0.0246 & 0.0312 & 0.0274 & 0.0406 & 0.0141 & 0.0174 \\
    CE$\rightarrow$DMPO & 0.0542 & 0.0924 & 0.0258 & 0.0353 & 0.0721 & 0.1109 & 0.0370 & 0.0468 & 0.0570 & 0.0892 & 0.0294 & 0.0375 & 0.0336 & 0.0542 & 0.0171 & 0.0223 \\
    CE$\rightarrow$S-DPO & 0.0628 & \cellcolor{gray!20}{0.1078} & \cellcolor{gray!20}0.0321 & \cellcolor{gray!20}0.0435 & 0.0770 & 0.1157 & 0.0395 & 0.0492 & 0.0606 & \cellcolor{gray!20}0.0929 & 0.0310 & 0.0392 & \cellcolor{gray!20}0.0389 & 0.0581 & 0.0202 & 0.0251 \\
    % \hline
    \cmidrule{1-17}
    \rowcolor[HTML]{e2edf8}
    \textbf{\ourmethod{}-Pointwise} & \underline{0.0667}$^{*}$ & \textbf{0.1105}$^{*}$ & \textbf{0.0337}$^{*}$ & \textbf{0.0447}$^{*}$ & \textbf{0.0815}$^{*}$ & \underline{0.1226}$^{*}$ & \textbf{0.0419}$^{*}$ & \underline{0.0523} & \underline{0.0637}$^{*}$ & \underline{0.0940} & \underline{0.0332}$^{*}$ & \underline{0.0409} & \underline{0.0411}$^{*}$ & \underline{0.0614}$^{*}$ & \underline{0.0214}$^{*}$ & \underline{0.0265} \\
    \rowcolor[HTML]{e2edf8}
    \textbf{\ourmethod{}-Pairwise} & \textbf{0.0671}$^{*}$ & \underline{0.1095} & \textbf{0.0337}$^{*}$ & \underline{0.0444} & \underline{0.0811}$^{*}$ & \textbf{0.1235}$^{*}$ & \underline{0.0418}$^{*}$ & \textbf{0.0525}$^{*}$ & \textbf{0.0639}$^{*}$ & \textbf{0.0964}$^{*}$ & \textbf{0.0339}$^{*}$ & \textbf{0.0421}$^{*}$ & \textbf{0.0412}$^{*}$ & \textbf{0.0623}$^{*}$ & \textbf{0.0218}$^{*}$ & \textbf{0.0271}$^{*}$ \\
    % \midrule
    % \cmidrule{2-17}
    \rowcolor[HTML]{e2edf8}
    \textit{\textbf{Max Improv.}} & \textit{\textbf{+5.17\%}} & \textit{\textbf{+2.50\%}} & \textit{\textbf{+4.98\%}} & \textit{\textbf{+2.76\%}} & \textit{\textbf{+5.19\%}} & \textit{\textbf{+4.48\%}} & \textit{\textbf{+3.92\%}} & \textit{\textbf{+3.75\%}} & \textit{\textbf{+4.58\%}} & \textit{\textbf{+3.77\%}} & \textit{\textbf{+6.60\%}} & \textit{\textbf{+5.78\%}} & \textit{\textbf{+5.91\%}} & \textit{\textbf{+6.31\%}} & \textit{\textbf{+5.83\%}} & \textit{\textbf{+5.45\%}} \\
    \midrule
    \rowcolor[HTML]{e5f1dd}
    \multicolumn{17}{c}{\textbf{\textit{Llama (Decoder Only)}}}\\
    % \cmidrule{1-17}
    CE & 0.0469 & 0.0732 & 0.0236 & 0.0302 & 0.0647 & \cellcolor{gray!20}{0.0985} & 0.0343 & 0.0428 & \cellcolor{gray!20}{0.0516} & \cellcolor{gray!20}{0.0770} & \cellcolor{gray!20}{0.0274} & \cellcolor{gray!20}{0.0338} & \cellcolor{gray!20}{0.0267} & \cellcolor{gray!20}{0.0398} & \cellcolor{gray!20}{0.0141} & \cellcolor{gray!20}{0.0174} \\
    CE$\rightarrow$DPO & 0.0477 & \cellcolor{gray!20}{0.0799} & 0.0228 & 0.0310 & 0.0620 & 0.0941 & 0.0323 & 0.0404 & 0.0364 & 0.0597 & 0.0172 & 0.0231 & 0.0171 & 0.0253 & 0.0089 & 0.0109 \\
    CE$\rightarrow$DMPO & 0.0442 & 0.0699 & 0.0228 & 0.0293 & 0.0649 & 0.0958 & 0.0353 & 0.0431 & 0.0457 & 0.0714 & 0.0251 & 0.0315 & 0.0251 & 0.0383 & 0.0130 & 0.0163 \\
    CE$\rightarrow$S-DPO & \cellcolor{gray!20}0.0493 & 0.0783 & \cellcolor{gray!20}0.0256 & \cellcolor{gray!20}0.0328 & \cellcolor{gray!20}0.0682 & 0.0975 & \cellcolor{gray!20}0.0371 & \cellcolor{gray!20}0.0444 & 0.0480 & 0.0700 & 0.0262 & 0.0318 & 0.0262 & 0.0372 & 0.0138 & 0.0166 \\
    \cmidrule{1-17}
    \rowcolor[HTML]{e5f1dd}
     \textbf{\ourmethod{}-Pointwise} & \textbf{0.0559}$^{*}$ & \underline{0.0871}$^{*}$ & \underline{0.0265}$^{*}$ & \underline{0.0343}$^{*}$ & \underline{0.0701} & \underline{0.1074}$^{*}$ & \underline{0.0379} & \underline{0.0473}$^{*}$ & \textbf{0.0564}$^{*}$ & \underline{0.0796}$^{*}$ & \textbf{0.0300}$^{*}$ & \underline{0.0359}$^{*}$ & \textbf{0.0289}$^{*}$ & \underline{0.0424}$^{*}$ & \textbf{0.0152}$^{*}$ & \textbf{0.0186}$^{*}$ \\
    \rowcolor[HTML]{e5f1dd}
    \textbf{\ourmethod{}-Pairwise} & \underline{0.0557}$^{*}$ & \textbf{0.0905}$^{*}$ & \textbf{0.0269}$^{*}$ & \textbf{0.0357}$^{*}$ & \textbf{0.0741}$^{*}$ & \textbf{0.1101}$^{*}$ & \textbf{0.0390}$^{*}$ & \textbf{0.0480}$^{*}$ & \underline{0.0548}$^{*}$ & \textbf{0.0835}$^{*}$ & \underline{0.0293}$^{*}$ & \textbf{0.0365}$^{*}$ & \underline{0.0287}$^{*}$ & \textbf{0.0425}$^{*}$ & \underline{0.0149}$^{*}$ & \underline{0.0184}$^{*}$ \\
    \rowcolor[HTML]{e5f1dd}
    \textit{\textbf{Max Improv.}} & \textit{\textbf{+13.39\%}} &	\textit{\textbf{+13.27\%}} &	\textit{\textbf{+5.08\%}} &	\textit{\textbf{+8.84\%}} &	\textit{\textbf{+8.65\%}} &	\textit{\textbf{+11.78\%}} &	\textit{\textbf{+5.12\%}} &	\textit{\textbf{+8.11\%}} &	\textit{\textbf{+9.30\%}} &	\textit{\textbf{+8.44\%}} &	\textit{\textbf{+9.49\%}} &	\textit{\textbf{+7.99\%}} &	\textit{\textbf{+8.24\%}} &	\textit{\textbf{+6.78\%}} &	\textit{\textbf{+7.80\%}} &	\textit{\textbf{+6.90\%}} \\
    \bottomrule
  \end{tabular}
  }
  \label{tab:main_res}
\end{table*}

In this section, we conduct extensive experiments to answer the following research questions:

\begin{itemize}[leftmargin=*]
    \item \textbf{RQ1:} How does our proposed \ourmethod{} perform compared to other state-of-the-art baselines?
    \item \textbf{RQ2:} How robust and effective is \ourmethod{} under different designs, training settings, and inference conditions?
    \item \textbf{RQ3:} Does our method truly bridge the training–inference gap by improving \textbf{prefix-level} predictions during beam search?
    \item \textbf{RQ4:} Is \ourmethod{} computationally efficient compared to baselines?
\end{itemize}

\begin{table}[!h]
\caption{Statistics of Datasets.}
  \centering
  \begin{tabular}{lcccc}
    \toprule
    \textbf{Dataset} & \textbf{\# Users} & \textbf{\# Items} & \textbf{\# Inter.} & \textbf{Density} \\
    \midrule
    Office & 4,905 & 2,420 & 53,258 & 0.45\% \\
    Grocery & 14,681 & 8,713 & 151,254 & 0.12\% \\
    Beauty & 22,363 & 12,101 & 198,502 & 0.07\% \\
    Yelp & 30,431 & 20,033 & 316,354 & 0.05\% \\
    \bottomrule
  \end{tabular}
  \label{tab:dataset}
\end{table}

\subsection{Experimental Settings}

\subsubsection{\textbf{Datasets and metrics}}

Following previous studies~\cite{rajput2023tiger, zheng2024lcrec, wang2024letter, liu2025ETEGRec}, we evaluate our approach on four public datasets of varying scales: \textbf{Office}\footnote{\href{https://jmcauley.ucsd.edu/data/amazon/links.html}{https://jmcauley.ucsd.edu/data/amazon/links.html}}, \textbf{Grocery}$^1$, \textbf{Beauty}$^1$, and \textbf{Yelp}\footnote{\href{https://www.yelp.com/dataset}{https://www.yelp.com/dataset}}.
Following the data processing protocols in previous work~\cite{wang2024letter, zheng2024lcrec}, we apply five-core filtering to discard users and items with fewer than five interactions, and adopt the \textit{leave-one-out} strategy for data splitting.
The statistics of the datasets after preprocessing are summarized in~\autoref{tab:dataset}.
% Evaluation
For evaluation, we employ top-\textit{K} \textit{Recall} (\textbf{R@K}) and \textit{Normalized Discounted Cumulative Gain} (\textbf{N@K})~\cite{jarvelin2002ndcg} with $K\in\{10,20\}$.

\subsubsection{\textbf{Baselines}}

To evaluate the effectiveness of our approach, we compare \ourmethod{} with different improved loss functions for generative recommendation:

\begin{itemize}[leftmargin=*]
\item \textbf{CE}~\cite{radford2018gpt1}, which optimizes the model with the standard cross-entropy loss.
\item \textbf{MSL}~\cite{wang2025msl}, which modifies CE loss by masking out invalid tokens in the softmax normalization.
\item \textbf{DPO}~\cite{rafailov2023DPO}, which aligns model outputs with preferences through a direct preference optimization objective.
\item \textbf{DMPO}~\cite{bai2024DMPO}, which extends DPO to multi-preference alignment by maximizing the positive while minimizing the average likelihood over $N$ sampled negatives.
\item \textbf{S-DPO}~\cite{chen2024SDPO}, which generalizes DPO to a multi-negative loss that weights negatives via softmax.
\end{itemize}

In line with prior studies~\cite{rajput2023tiger, liu2025ETEGRec, zheng2024lcrec}, we conducted experiments on two backbone GR models: TIGER~\cite{rajput2023tiger} (encoder–decoder architecture) and Llama~\cite{dubey2024llama} (decoder-only architecture), both instantiated with comparable parameter sizes ($\approx$0.01B).

\subsubsection{\textbf{Implementation Details}}

% tokenization
For tokenization in generative recommendation, we follow prior studies~\cite{wang2024letter, zheng2024lcrec} and first extract semantic embeddings using \texttt{Llama-3.1-8B-Instruct}~\cite{dubey2024llama}. We then apply a 4-level RQ-K-means quantizer~\cite{zhou2025onerec} to obtain semantic codes for all items.
For sequential recommendation, we follow previous work~\cite{rajput2023tiger, wang2024letter} and truncate user history to the most recent 20 items.
We adopt a learning rate of 5e-4 with 1\% warmup and a cosine decay scheduler. 
AdamW~\cite{loshchilov2017decoupled} is used as the default optimizer.
Training is conducted for up to 200 epochs with early stopping patience set to 20. The batch size is 1024, implemented as 128 samples per device with 8 steps of gradient accumulation.
All models are trained on 1 NVIDIA A100-SXM4-40GB.
For model inference, we use beam search with a beam size of 20.

% hyperparameter
For our method, the weight $\beta$ of the $\mathcal{L}_\text{prefix}$ is tuned over \{0.05, 0.1, 0.2, 0.3, 0.4\}. $\eta$ is tuned over \{5e-6, 1e-5, 3e-5, 5e-5, 1e-4, 5e-4\}.
% Negative Sampling
To ensure fairness, we adopt uniform random negative sampling and fix the number of negative samples to 100 for all methods. Exploring more advanced sampling strategies is left for future work.
All baseline hyperparameters are tuned following suggestions from their original papers.
% gpu & codes
Full configurations, including hyperparameters, model checkpoints, and scripts, are available in our code repository\footnote{\href{https://github.com/yuyq18/APAO}{https://github.com/yuyq18/APAO}} for reproducibility.

\subsection{RQ1: Overall Performance}

\autoref{tab:main_res} presents the top-$K$ recommendation performance of \ourmethod{} and various baselines across four datasets, using two representative generative recommendation backbones: TIGER (encoder–decoder) and Llama (decoder-only). Below are some key observations:

\begin{itemize}[leftmargin=*]
\item \textbf{Overall superiority of \ourmethod{}.}
Across all datasets, backbones, and baselines, \ourmethod{} consistently achieves the best performance on both Recall and NDCG.
By incorporating prefix-aware supervision that aligns the training objective with the beam-search generation process, \ourmethod{} effectively mitigates the training–inference inconsistency and achieves significant performance gains.
Moreover, consistent improvements across the \textit{Encoder–Decoder} and the \textit{Decoder-Only} architectures further demonstrate the generality and robustness of the proposed framework.

\item \textbf{Advantage over preference optimization baselines.}
Compared with recent preference-optimization approaches such as DPO, DMPO, and S-DPO, \ourmethod{} delivers notable and consistent improvements.
Unlike these methods that rely solely on the final generated items for optimization, \ourmethod{} introduces consistent supervision across all prefixes, ensuring that each partial generation step is guided toward optimal beam expansion.

\item \textbf{Effectiveness of both pointwise and pairwise variants.}
Both variants of \ourmethod{} demonstrate robust performance with complementary strengths: the pointwise version excels in stability and efficiency, whereas the pairwise version generally yields superior ranking quality by explicitly modeling token-level preferences. 
This confirms the versatility of our framework across different supervision formulations.

\end{itemize}

\subsection{RQ2: In-depth Analysis}

\begin{figure}[h]
    \centering
    \includegraphics[width=1.0\linewidth]{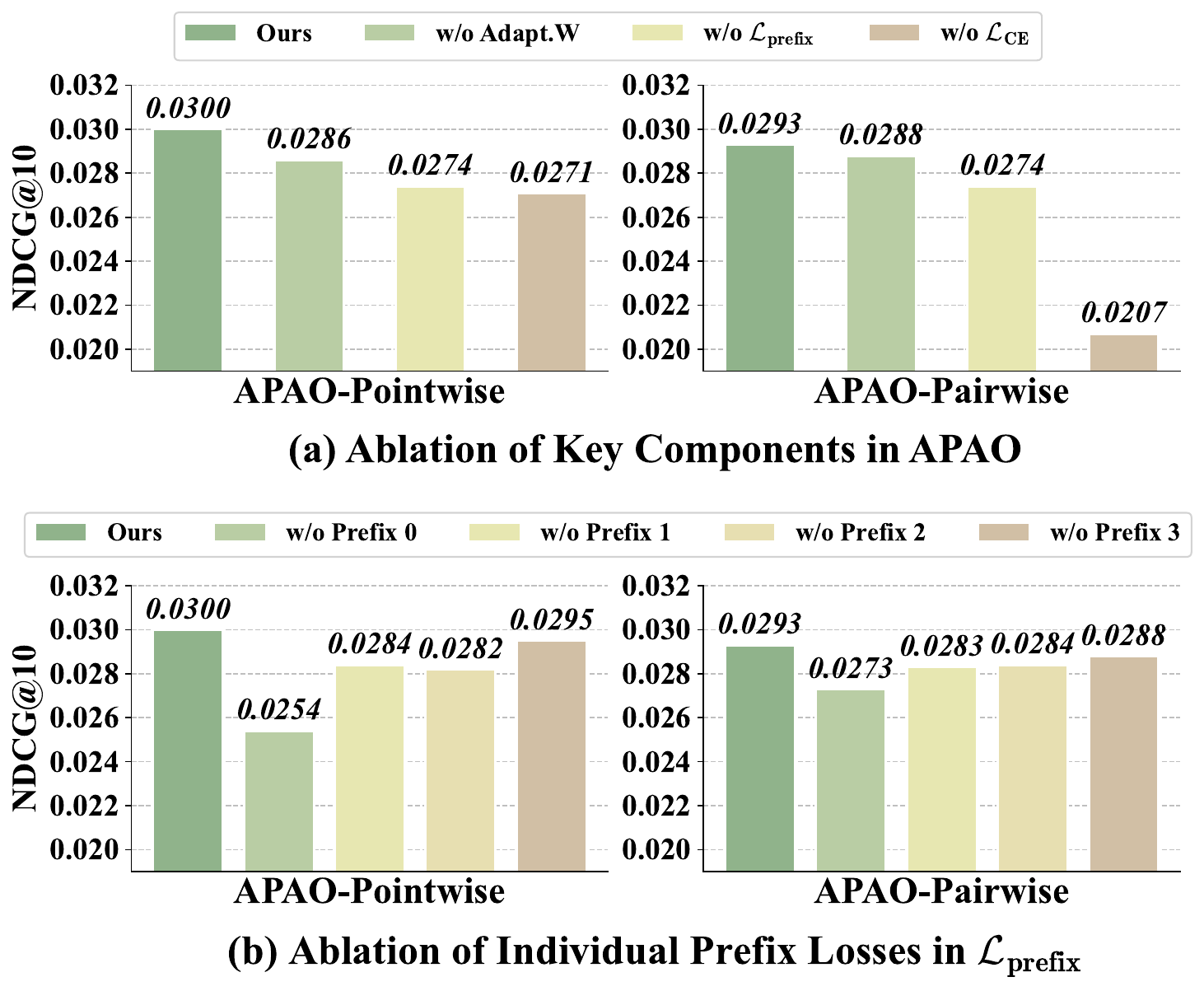}
    \caption{Ablation analysis on the Beauty dataset using the Llama backbone. Similar trends are observed across other datasets.}
    \label{fig:ablation}
\end{figure}

\subsubsection{\textbf{Ablation Study}}
To further examine the effectiveness of individual components in \ourmethod{}, we perform ablation experiments on the Beauty dataset using the Llama backbone.

\begin{itemize}[leftmargin=*]
\item \textbf{Per-Module Ablation.}
\autoref{fig:ablation}(a) reports the impact of removing key components. 
% Overall
We observe that removing any single component leads to a performance drop, which confirms that the improvements of \ourmethod{} arise from the joint contribution of all modules.
In particular, the absence of adaptive weighting reduces the model’s ability to balance informative supervision, while excluding $\mathcal{L}_\text{CE}$ or $\mathcal{L}_\text{prefix}$ significantly weakens the overall optimization.

\item \textbf{Prefix Ablation.}
\autoref{fig:ablation}(b) evaluates the impact of removing loss supervision at different prefix levels.
% Overall
Results confirm that all prefixes contribute to performance; however, removing early prefixes (e.g., Prefix~0) is significantly more detrimental, and this impact progressively decreases for later prefixes.
This highlights the critical role of early supervision for guiding the beam search to the correct target.

\begin{figure}[h]
    \centering
    \includegraphics[width=1.0\linewidth]{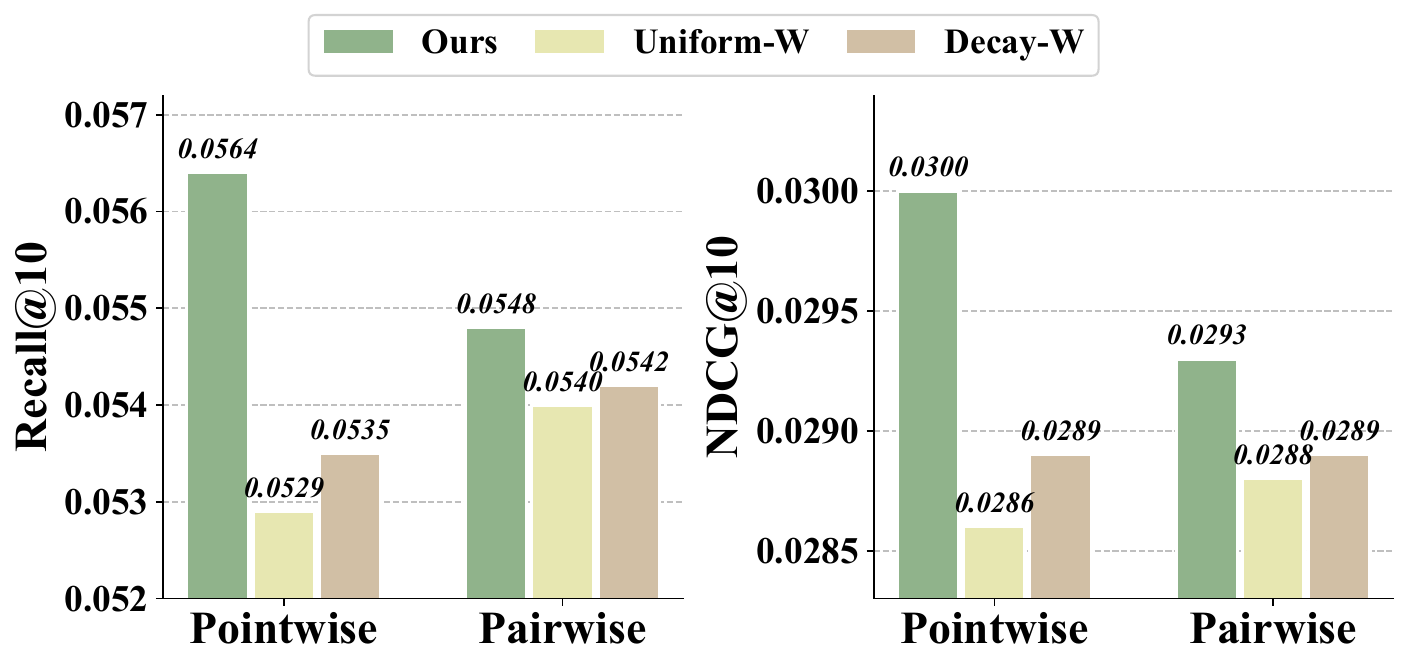}
    \caption{Weighting strategy analysis on the Beauty dataset using the Llama backbone. Similar trends are observed across other datasets.}
    \label{fig:weighting}
\end{figure}

\item \textbf{Weighting Strategy Ablation.}
We further compare different prefix weighting strategies, including uniform weighting and fixed exponentially decaying weighting ($w_t \propto e^{-\lambda t}$).
Notably, the ``w/o Adapt.W'' variant in Fig.~\ref{fig:ablation} is essentially equivalent to uniform weighting over prefixes.
As shown in Fig.~\ref{fig:weighting}, adaptive weighting consistently achieves the best performance for both pointwise and pairwise variants.
Compared with fixed weighting strategies, adaptive weighting dynamically focuses on the currently weakest prefix, which better matches the prefix-level pruning behavior of beam search.

\end{itemize}

\begin{table}[h]
\centering
\small
\caption{Comparison of different training paradigms on the Office and Beauty datasets in terms of NDCG@10.
\textbf{Bold} indicates the overall best result, while \underline{underlined} marks the best within each group.}
\resizebox{0.48\textwidth}{!}{
\begin{tabular}{lcccc}
\toprule
\multirow{2}{*}{\textbf{Methods}} & \multicolumn{2}{c}{\textbf{Office}} & \multicolumn{2}{c}{\textbf{Beauty}} \\
\cmidrule(r){2-3}\cmidrule(r){4-5}
 & \textbf{TIGER} & \textbf{Llama} & \textbf{TIGER} & \textbf{Llama}  \\
\midrule
1-Stage [$\mathcal{L}_{\text{CE}}+ \lambda \mathcal{L}_{\text{S-DPO}}$]     & 0.0116 & 0.0089 & 0.0167 & 0.0146 \\
\textbf{1-Stage [\ourmethod{}-Pointwise]}     & \textbf{\underline{0.0337}} & 0.0265 & 0.0332 & \underline{0.0300} \\
\textbf{1-Stage [\ourmethod{}-Pairwise]}     & \textbf{\underline{0.0337}}& \underline{0.0269} & \textbf{\underline{0.0339}} & 0.0293 \\
\midrule
2-Stage [SDPO] & 0.0321 & 0.0256 & 0.0310 & 0.0262 \\
2-Stage [$\mathcal{L}_{\text{CE}} \rightarrow \mathcal{L}_{\text{point}}$ (Ours)]   & \underline{0.0323} & \textbf{\underline{0.0278}} & 0.0310 & \textbf{\underline{0.0308}} \\
2-Stage [$\mathcal{L}_{\text{CE}} \rightarrow\mathcal{L}_{\text{pair}}$ (Ours)]   & 0.0320 & 0.0260 & \underline{0.0312} & 0.0281 \\
\bottomrule
\end{tabular}
}
\label{tab:onestage_twostage}
\end{table}

\subsubsection{\textbf{One-Stage vs. Two-Stage Training}}

While standard alignment baselines (e.g., S-DPO) typically adopt a \textbf{two-stage} pipeline (SFT$\rightarrow \mathcal{L}_{\text{S-DPO}}$), our method is naturally formulated under a unified \textbf{one-stage} paradigm via $\mathcal{L}_{\text{unified}} = \mathcal{L}_{\text{CE}} + \beta \mathcal{L}_{\text{prefix}}$.
To ensure fair comparison and isolate the impact of the training strategy, we conduct cross-experiments by implementing S-DPO in a one-stage setting and adapting our method to a two-stage setting.
\autoref{tab:onestage_twostage} presents the results on the Office and Beauty datasets with both the TIGER and Llama backbones. The results reveal the following key observations:
\begin{itemize}[leftmargin=*]

\item \textbf{Consistent effectiveness across training paradigms.}
Our method achieves superior performance under the one-stage setting for TIGER and competitive results under the two-stage setting for Llama. This indicates that our proposed ranking objectives are highly versatile and can be effectively adapted to different training strategies to maximize model performance.
\item \textbf{Prefix-aware optimization consistently outperforms S-DPO.}
Under identical conditions, our method (both pairwise and pointwise) surpasses S-DPO across all metrics, demonstrating the advantage of our prefix-aware optimization design.
\end{itemize}

\subsubsection{\textbf{Robustness on Tokenization}}

Following prior work, we primarily adopt the standard RQ tokenizer for all methods to ensure fair comparison.
To further evaluate the robustness of \ourmethod{} under different tokenization qualities and structures, we additionally consider two alternative tokenizers: \textit{Suboptimal RQ}, which introduces random perturbations to 10\% of the codes at 2 out of 4 levels, and \textit{OPQ}, which uses a 4-code product quantization scheme.

\begin{figure}[h]
    \centering
    \includegraphics[width=0.9\linewidth]{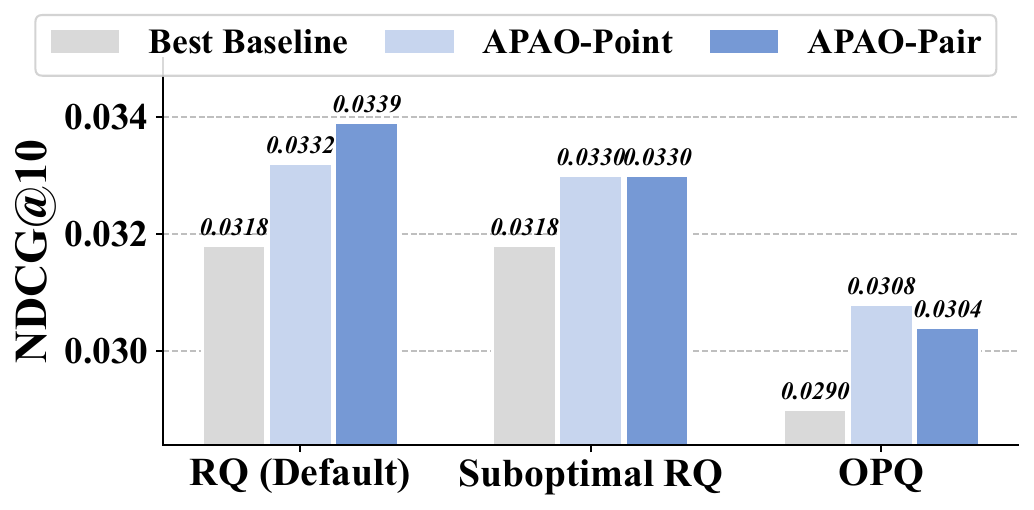}
    \caption{Robustness analysis on the Beauty dataset using the TIGER backbone. Similar trends are observed across other datasets.}
    \label{fig:tokenization}
\end{figure}

As shown in \autoref{fig:tokenization}, both pointwise and pairwise variants of \ourmethod{} consistently outperform the best baseline across all tokenizer settings.
Although the absolute performance varies with tokenizer quality, the relative improvements remain stable, demonstrating that the effectiveness of \ourmethod{} is not tied to a specific tokenizer design.
These results suggest that prefix-aware optimization generalizes well across different semantic tokenization strategies.

\begin{figure}[h]
    \centering
    \includegraphics[width=1.0\linewidth]{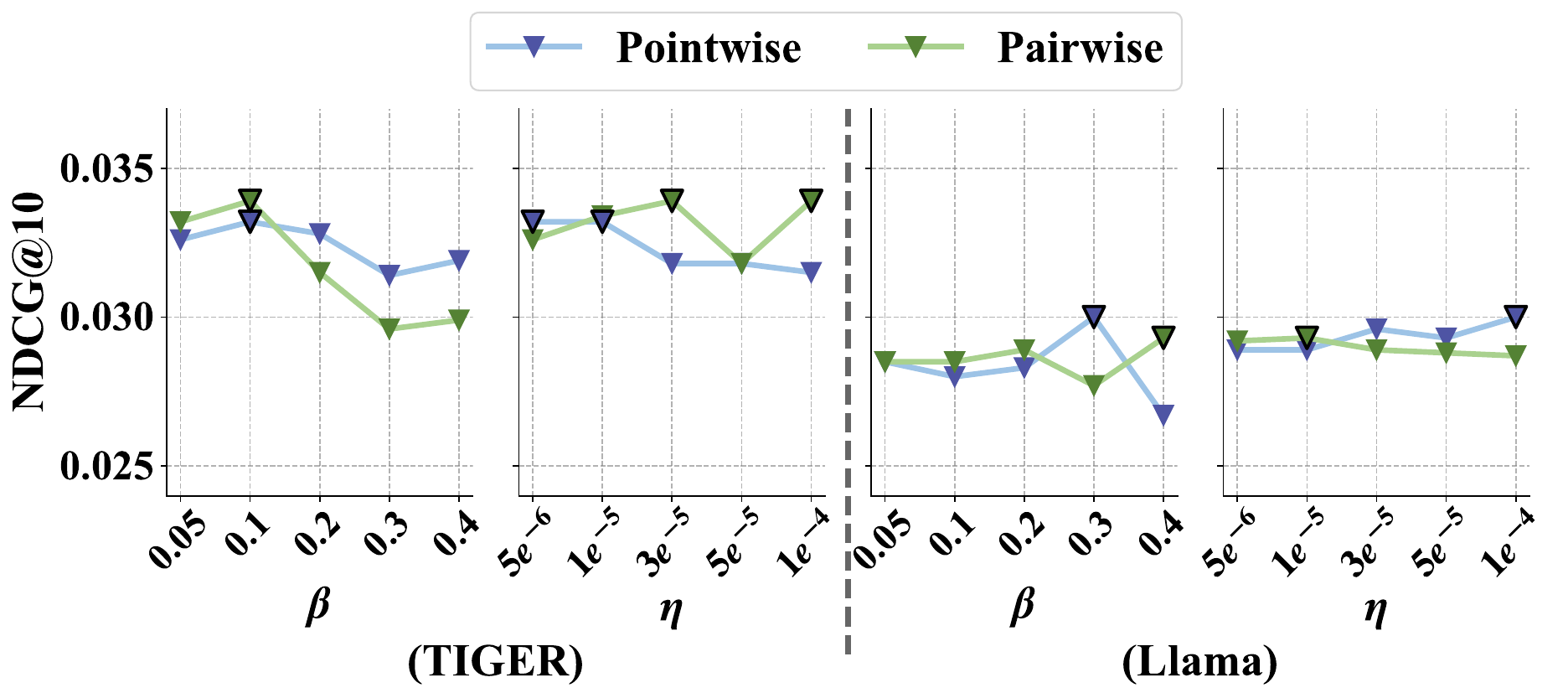}
    \caption{Parameter sensitivity analysis of \ourmethod{}-pointwise and \ourmethod{}-pairwise with respect to hyperparameters $\beta$ and $\eta$ on the Beauty dataset. The experiments are conducted on two generative backbones: TIGER (left) and Llama (right).}
    \label{fig:param_sensitivity}
\end{figure}

\subsubsection{\textbf{Hyperparameter Analysis}}

We conduct a sensitivity analysis of key hyperparameters in \ourmethod{} to evaluate their influence. The results on the Beauty dataset are illustrated in~\autoref{fig:param_sensitivity}, while the detailed optimal configurations for all datasets are provided in~\autoref{tab:reproduction} (\textbf{\autoref{appendix:implement}}).

\begin{itemize}[leftmargin=*]
    \item \textbf{Impact of Prefix-aware Loss Weight ($\beta$).} The parameter $\beta$ determines the contribution of the auxiliary prefix loss to the total objective. 
    As shown in \autoref{fig:param_sensitivity}, while the optimal $\beta$ varies across different backbones, we empirically find that values within the range of $[0.1, 0.4]$ generally yield favorable performance.
    \item \textbf{Impact of Adaptive Weighting Factor ($\eta$).} 
    The hyperparameter $\eta$ controls the sensitivity of the adaptive weighting mechanism (Eq.~(\ref{eq:weight_update})).
    As shown in \autoref{fig:param_sensitivity}, our method maintains relatively robust performance across a range of $\eta$. 
    Additionally, we empirically observed that larger datasets tend to benefit from a smaller $\eta$, which enhances stability during optimization.
\end{itemize}

\begin{figure}[h]
    \centering
    \includegraphics[width=1.0\linewidth]{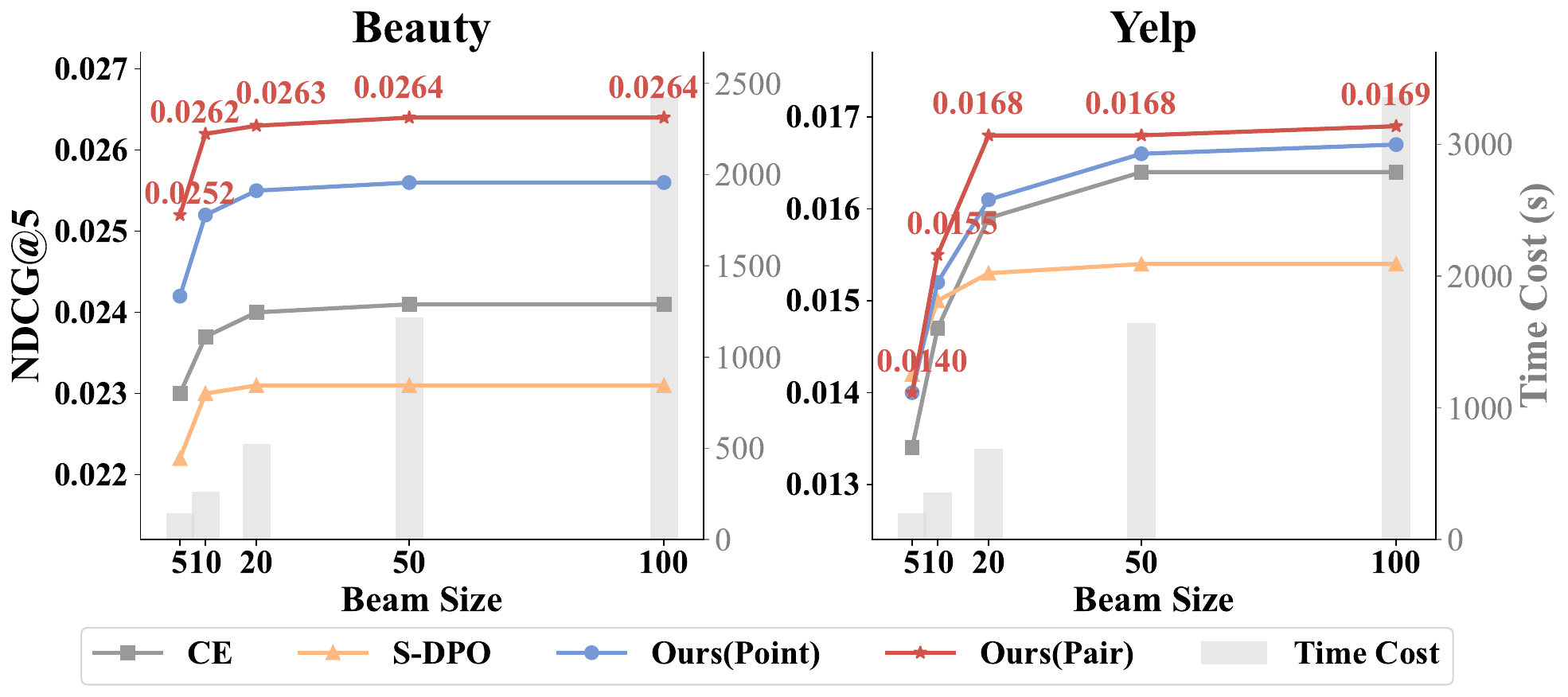}
    \caption{Scalability analysis of beam size $K$ on the Beauty and Yelp datasets.}
    \label{fig:beam_size}
\end{figure}

\subsubsection{\textbf{Scalability w.r.t. Beam Size}}

Finally, we investigate the scalability of \ourmethod{} by varying the inference beam size $K$ on the Beauty and Yelp datasets, simultaneously recording the ranking performance (NDCG@5) and total time cost.
As shown in \autoref{fig:beam_size}, \ourmethod{} (both pointwise and pairwise variants) consistently outperforms baselines across all beam sizes. Notably, our method with $K=20$ achieves comparable or better NDCG@5 than baselines with $K=100$.
This confirms both the \textbf{effectiveness} of our approach in limited search steps and its \textbf{robustness} over varying beam sizes.

\subsection{RQ3: Prefix-level Performance}

To evaluate whether \ourmethod{} effectively bridges the training–inference gap at the granularity of beam expansion, we analyze \textbf{prefix-level \textit{Recall@20}} over the four semantic code levels produced by the 4-level tokenizer. Specifically, prefix indices 0–3 correspond to progressively longer semantic prefixes, with prefix 3 representing the complete item.

\autoref{fig:prefix_recall} presents the results on the Office and Beauty datasets by comparing the best variant of our method against the best-performing baseline, detailing both the absolute recall and the relative improvement.

\begin{figure}[h]
    \centering
    \includegraphics[width=1.0\linewidth]{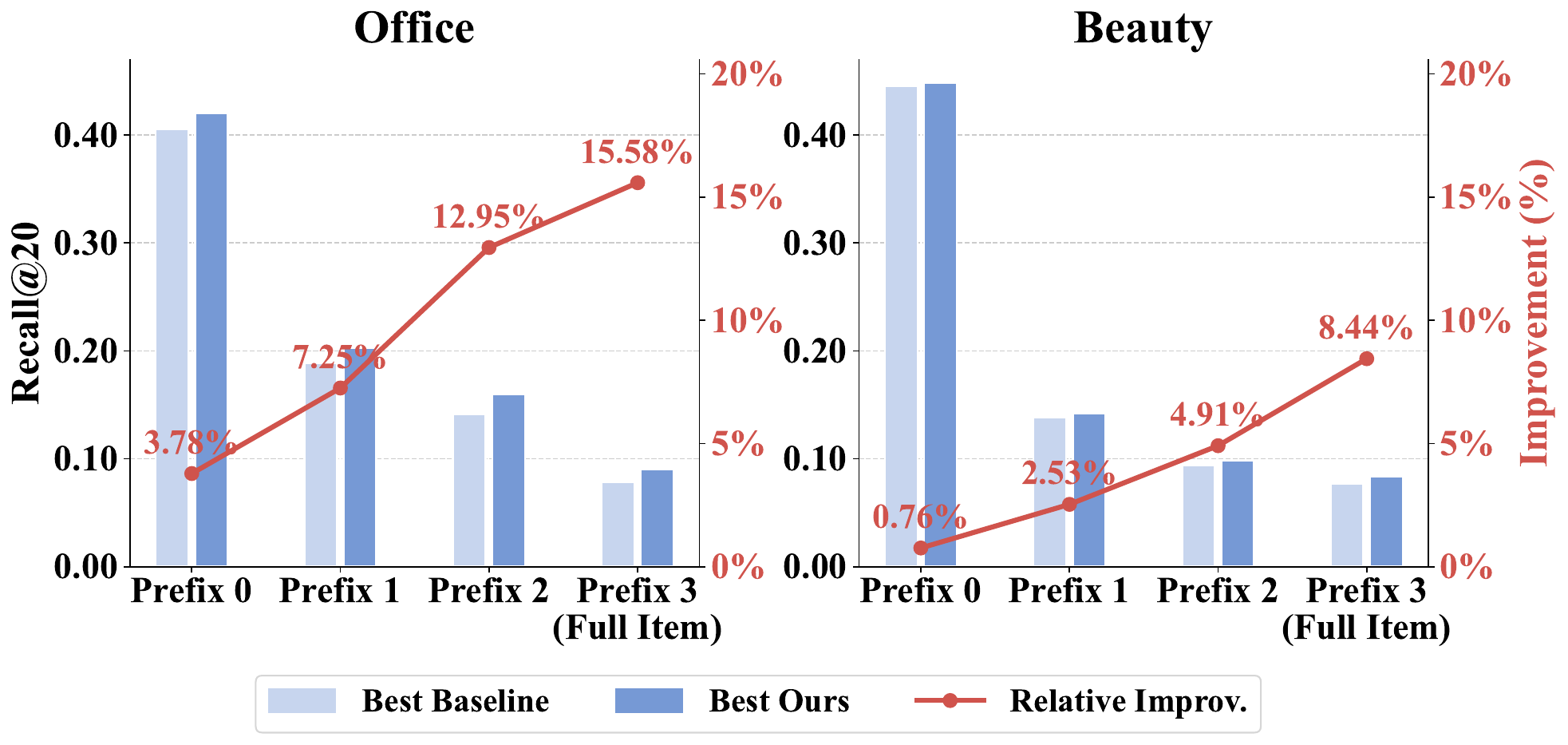}
    \caption{Prefix-level Recall@20 and relative improvement on the Office and Beauty datasets. The left axis represents absolute recall, and the right axis denotes the relative improvement over the best baseline across prefixes 0–3 (prefix 3 corresponds to the complete item).}
    \label{fig:prefix_recall}
\end{figure}

In terms of absolute recall, \ourmethod{} consistently outperforms the best baseline across all prefix lengths. The relative improvement also grows steadily as the prefix length increases, suggesting that \ourmethod{} more effectively prevents the ground-truth item from being pruned as the beam expands.
These results indicate that \ourmethod{} enhances both intermediate and final top-$K$ performance during beam search, aligning the training objective with the inference procedure and alleviating the training–inference inconsistency.

\subsection{RQ4: Efficiency Analysis}
\label{sec: efficiency}
Consistent with the complexity analysis in Section~\ref{sec: Time_Complexity_Analysis}, we empirically evaluate the training efficiency of \ourmethod{} against representative baselines.
As shown in \autoref{tab:efficiency_time}, we report the average training time per epoch (\textbf{Time}) and the total epochs to convergence (\textbf{Epochs}).

\begin{table}[h]
    \centering
    \tabcolsep=0.08cm
    \caption{Training efficiency comparison on the TIGER backbone. The results are presented in the format of \textbf{Time per Epoch (s) / Total Epochs to Convergence}.
}
    \begin{tabular}{llcccc}
        \toprule
        \multirow{2}{*}{\textbf{Type}} & \multirow{2}{*}{\textbf{Methods}} & \multicolumn{4}{c}{\textbf{Time (s) | Epochs}} \\
        \cmidrule(r){3-6}
        & & \textbf{Office} & \textbf{Grocery} & \textbf{Beauty} & \textbf{Yelp}\\
        \midrule
        \multirow{2}{*}{Pointwise} & CE & 15 | 45  & 39 | 63 & 53 | 74 & 66 | 123 \\
        & \ourmethod{} & 15 | 45 & 35 | 63 & 54 | 83 & 79 | 165 \\
        \midrule
        \multirow{2}{*}{Pairwise} & S-DPO & 199 | 93 & 572 | 106 & 693 | 109 & 1175 | 159 \\
        & \ourmethod{} & 163 | 58 & 456 | 67 & 637 | 119 & 946 | 144 \\
        \bottomrule
    \end{tabular}
    \label{tab:efficiency_time}
\end{table}

Specifically, for our \textbf{pointwise} variant, the added prefix-aware computation is efficiently implemented and maintains training efficiency comparable to the CE baseline.
For the \textbf{pairwise} variant, employing 100 negative samples inevitably incurs higher costs than pointwise methods. However, \ourmethod{}-Pairwise trains notably faster than S-DPO, as it bypasses the additional inference required by the reference model.
Overall, \ourmethod{} effectively improves model performance while preserving practical computational efficiency.

\subsection{Industrial Deployment}

\ourmethod{}-Pointwise has recently been deployed in a real-world industrial recommendation system on the WeChat Official Accounts Platform.
Online A/B testing against the production CE baseline demonstrates consistent improvements across multiple metrics, including +0.9\% pCTR, +0.706\% image CTR, +0.907\% image clicks per user, and +0.205\% image impressions per user.

\section{Conclusion}
In this work, we addressed the training–inference inconsistency in generative recommendation caused by the prefix-level constraints introduced by beam search decoding. To tackle this issue, we proposed the Adaptive Prefix-Aware Optimization (APAO) framework, which introduces prefix-level optimization during training to better align the model with its inference process. APAO includes both pointwise and pairwise optimization objectives that capture complementary learning perspectives and an adaptive worst-prefix optimization strategy that dynamically emphasizes the most challenging prefixes during training. Theoretical and empirical analyses demonstrate that APAO effectively enhances alignment with beam search decoding while introducing no significant computational overhead. Extensive experiments confirm that APAO consistently improves ranking performance and alleviates the training–inference gap.
We hope this work can inspire future research toward inference-aware training for generative recommendation.

\begin{acks}
This work is supported by the National Natural Science Foundation of China (Grant No. 62372260), the National Key Research and Development Program of China (Grant No. 2024YFC3307403), and Quancheng Laboratory of China (Grant No. QCL20250105). Weizhi Ma is also sponsored by the Beijing Nova Program.
\end{acks}

\bibliographystyle{ACM-Reference-Format}
\balance
\bibliography{sample-base}

@String{Computer = "{IEEE} Computer" }

@String{Chelsea = "Chelsea" }

@article{sun2023GenRetrie,
  title={Learning to tokenize for generative retrieval},
  author={Sun, Weiwei and Yan, Lingyong and Chen, Zheng and Wang, Shuaiqiang and Zhu, Haichao and Ren, Pengjie and Chen, Zhumin and Yin, Dawei and Rijke, Maarten and Ren, Zhaochun},
  journal={Advances in Neural Information Processing Systems},
  volume={36},
  pages={46345--46361},
  year={2023}
}

@article{yang2023AutoSearchIndexer,
  title={Auto search indexer for end-to-end document retrieval},
  author={Yang, Tianchi and Song, Minghui and Zhang, Zihan and Huang, Haizhen and Deng, Weiwei and Sun, Feng and Zhang, Qi},
  journal={arXiv preprint arXiv:2310.12455},
  year={2023}
}

@article{zeng2022glm,
  title={Glm-130b: An open bilingual pre-trained model},
  author={Zeng, Aohan and Liu, Xiao and Du, Zhengxiao and Wang, Zihan and Lai, Hanyu and Ding, Ming and Yang, Zhuoyi and Xu, Yifan and Zheng, Wendi and Xia, Xiao and others},
  journal={arXiv preprint arXiv:2210.02414},
  year={2022}
}

@article{gpt42023technical,
  title={GPT-4 Technical Report},
  author={OpenAI},
  journal={arXiv preprint arXiv:2303.08774},
  year={2023}
}

@article{dubey2024llama,
  title={The llama 3 herd of models},
  author={Grattafiori, Aaron and Dubey, Abhimanyu and Jauhri, Abhinav and Pandey, Abhinav and Kadian, Abhishek and Al-Dahle, Ahmad and Letman, Aiesha and Mathur, Akhil and Schelten, Alan and Vaughan, Alex and others},
  journal={arXiv preprint arXiv:2407.21783},
  year={2024}
}

@inproceedings{zheng2024lcrec,
  title={Adapting large language models by integrating collaborative semantics for recommendation},
  author={Zheng, Bowen and Hou, Yupeng and Lu, Hongyu and Chen, Yu and Zhao, Wayne Xin and Chen, Ming and Wen, Ji-Rong},
  booktitle={2024 IEEE 40th International Conference on Data Engineering (ICDE)},
  pages={1435--1448},
  year={2024},
  organization={IEEE}
}

@article{rajput2023tiger,
  title={Recommender systems with generative retrieval},
  author={Rajput, Shashank and Mehta, Nikhil and Singh, Anima and Hulikal Keshavan, Raghunandan and Vu, Trung and Heldt, Lukasz and Hong, Lichan and Tay, Yi and Tran, Vinh and Samost, Jonah and others},
  journal={Advances in Neural Information Processing Systems},
  volume={36},
  pages={10299--10315},
  year={2023}
}

@article{zhou2025onerec,
  title={OneRec Technical Report},
  author={Zhou, Guorui and Deng, Jiaxin and Zhang, Jinghao and Cai, Kuo and Ren, Lejian and Luo, Qiang and Wang, Qianqian and Hu, Qigen and Huang, Rui and Wang, Shiyao and others},
  journal={arXiv preprint arXiv:2506.13695},
  year={2025}
}

@article{zhou2025onerecv2,
  title={OneRec-V2 Technical Report},
  author={Zhou, Guorui and Hu, Hengrui and Cheng, Hongtao and Wang, Huanjie and Deng, Jiaxin and Zhang, Jinghao and Cai, Kuo and Ren, Lejian and Ren, Lu and Yu, Liao and others},
  journal={arXiv preprint arXiv:2508.20900},
  year={2025}
}

@inproceedings{wang2024letter,
  title={Learnable item tokenization for generative recommendation},
  author={Wang, Wenjie and Bao, Honghui and Lin, Xinyu and Zhang, Jizhi and Li, Yongqi and Feng, Fuli and Ng, See-Kiong and Chua, Tat-Seng},
  booktitle={Proceedings of the 33rd ACM International Conference on Information and Knowledge Management},
  pages={2400--2409},
  year={2024}
}

@inproceedings{liu2025ETEGRec,
  title={Generative recommender with end-to-end learnable item tokenization},
  author={Liu, Enze and Zheng, Bowen and Ling, Cheng and Hu, Lantao and Li, Han and Zhao, Wayne Xin},
  booktitle={Proceedings of the 48th International ACM SIGIR Conference on Research and Development in Information Retrieval},
  pages={729--739},
  year={2025}
}

@inproceedings{wang2025msl,
  title={Msl: Not all tokens are what you need for tuning llm as a recommender},
  author={Wang, Bohao and Liu, Feng and Chen, Jiawei and Lou, Xingyu and Zhang, Changwang and Wang, Jun and Sun, Yuegang and Feng, Yan and Chen, Chun and Wang, Can},
  booktitle={Proceedings of the 48th International ACM SIGIR Conference on Research and Development in Information Retrieval},
  pages={1912--1922},
  year={2025}
}

@article{chen2024SDPO,
  title={On softmax direct preference optimization for recommendation},
  author={Chen, Yuxin and Tan, Junfei and Zhang, An and Yang, Zhengyi and Sheng, Leheng and Zhang, Enzhi and Wang, Xiang and Chua, Tat-Seng},
  journal={Advances in Neural Information Processing Systems},
  volume={37},
  pages={27463--27489},
  year={2024}
}

@article{rafailov2023DPO,
  title={Direct preference optimization: Your language model is secretly a reward model},
  author={Rafailov, Rafael and Sharma, Archit and Mitchell, Eric and Manning, Christopher D and Ermon, Stefano and Finn, Chelsea},
  journal={Advances in neural information processing systems},
  volume={36},
  pages={53728--53741},
  year={2023}
}

@inproceedings{renzhaochun2025constrained,
  title={Constrained Auto-Regressive Decoding Constrains Generative Retrieval},
  author={Wu, Shiguang and Ren, Zhaochun and Xin, Xin and Yang, Jiyuan and Zhang, Mengqi and Chen, Zhumin and de Rijke, Maarten and Ren, Pengjie},
  booktitle={Proceedings of the 48th International ACM SIGIR Conference on Research and Development in Information Retrieval},
  pages={2429--2440},
  year={2025}
}

@inproceedings{bai2024DMPO,
  title={Aligning large language model with direct multi-preference optimization for recommendation},
  author={Bai, Zhuoxi and Wu, Ning and Cai, Fengyu and Zhu, Xinyi and Xiong, Yun},
  booktitle={Proceedings of the 33rd ACM International Conference on Information and Knowledge Management},
  pages={76--86},
  year={2024}
}

@inproceedings{sun2019bert4rec,
  title={BERT4Rec: Sequential recommendation with bidirectional encoder representations from transformer},
  author={Sun, Fei and Liu, Jun and Wu, Jian and Pei, Changhua and Lin, Xiao and Ou, Wenwu and Jiang, Peng},
  booktitle={Proceedings of the 28th ACM international conference on information and knowledge management},
  pages={1441--1450},
  year={2019}
}

@inproceedings{kang2018sasrec,
  title={Self-attentive sequential recommendation},
  author={Kang, Wang-Cheng and McAuley, Julian},
  booktitle={2018 IEEE international conference on data mining (ICDM)},
  pages={197--206},
  year={2018},
  organization={IEEE}
}

@article{hidasi2015gru4rec,
  title={Session-based recommendations with recurrent neural networks},
  author={Hidasi, Bal{\'a}zs and Karatzoglou, Alexandros and Baltrunas, Linas and Tikk, Domonkos},
  journal={arXiv preprint arXiv:1511.06939},
  year={2015}
}

@inproceedings{tang2018Caser,
  title={Personalized top-n sequential recommendation via convolutional sequence embedding},
  author={Tang, Jiaxi and Wang, Ke},
  booktitle={Proceedings of the eleventh ACM international conference on web search and data mining},
  pages={565--573},
  year={2018}
}

@article{radford2018gpt1,
  title={Improving language understanding by generative pre-training},
  author={Radford, Alec and Narasimhan, Karthik and Salimans, Tim and Sutskever, Ilya and others},
  year={2018},
  publisher={San Francisco, CA, USA}
}

@inproceedings{hua2023index,
  title={How to index item ids for recommendation foundation models},
  author={Hua, Wenyue and Xu, Shuyuan and Ge, Yingqiang and Zhang, Yongfeng},
  booktitle={Proceedings of the Annual International ACM SIGIR Conference on Research and Development in Information Retrieval in the Asia Pacific Region},
  pages={195--204},
  year={2023}
}

@inproceedings{lee2022rqvae,
  title={Autoregressive image generation using residual quantization},
  author={Lee, Doyup and Kim, Chiheon and Kim, Saehoon and Cho, Minsu and Han, Wook-Shin},
  booktitle={Proceedings of the IEEE/CVF conference on computer vision and pattern recognition},
  pages={11523--11532},
  year={2022}
}

@inproceedings{liu2024multi,
  title={Multi-behavior generative recommendation},
  author={Liu, Zihan and Hou, Yupeng and McAuley, Julian},
  booktitle={Proceedings of the 33rd ACM International Conference on Information and Knowledge Management},
  pages={1575--1585},
  year={2024}
}

@article{zhai2025multimodal,
  title={Multimodal Quantitative Language for Generative Recommendation},
  author={Zhai, Jianyang and Mai, Zi-Feng and Wang, Chang-Dong and Yang, Feidiao and Zheng, Xiawu and Li, Hui and Tian, Yonghong},
  journal={arXiv preprint arXiv:2504.05314},
  year={2025}
}

@InProceedings{houactionpiece,
  title = 	 {{A}ction{P}iece: Contextually Tokenizing Action Sequences for Generative Recommendation},
  author =       {Hou, Yupeng and Ni, Jianmo and He, Zhankui and Sachdeva, Noveen and Kang, Wang-Cheng and Chi, Ed H. and Mcauley, Julian and Cheng, Derek Zhiyuan},
  booktitle = 	 {Proceedings of the 42nd International Conference on Machine Learning},
  pages = 	 {24004--24024},
  year = 	 {2025},
  editor = 	 {Singh, Aarti and Fazel, Maryam and Hsu, Daniel and Lacoste-Julien, Simon and Berkenkamp, Felix and Maharaj, Tegan and Wagstaff, Kiri and Zhu, Jerry},
  volume = 	 {267},
  series = 	 {Proceedings of Machine Learning Research},
  month = 	 {13--19 Jul},
  publisher =    {PMLR},
  url = 	 {https://proceedings.mlr.press/v267/hou25f.html},
}

@inproceedings{lin-efficient,
 author = {Lin, Xinyu and Yang, Chaoqun and Wang, Wenjie and Li, Yongqi and Du, Cunxiao and Feng, Fuli and Ng, See-Kiong and Chua, Tat-Seng},
 booktitle = {International Conference on Learning Representations},
 editor = {Y. Yue and A. Garg and N. Peng and F. Sha and R. Yu},
 pages = {91672--91697},
 title = {Efficient Inference for Large Language Model-based Generative Recommendation},
 url = {https://proceedings.iclr.cc/paper_files/paper/2025/file/e4bf5c3245fd92a4554a16af9803b757-Paper-Conference.pdf},
 volume = {2025},
 year = {2025}
}

@inproceedings{lin2025order,
  title={Order-agnostic Identifier for Large Language Model-based Generative Recommendation},
  author={Lin, Xinyu and Shi, Haihan and Wang, Wenjie and Feng, Fuli and Wang, Qifan and Ng, See-Kiong and Chua, Tat-Seng},
  booktitle={Proceedings of the 48th international ACM SIGIR conference on research and development in information retrieval},
  pages={1923--1933},
  year={2025}
}

@inproceedings{yang2025earn,
  title={EARN: Efficient Inference Acceleration for LLM-based Generative Recommendation by Register Tokens},
  author={Yang, Chaoqun and Lin, Xinyu and Wang, Wenjie and Li, Yongqi and Sun, Teng and Han, Xianjing and Chua, Tat-Seng},
  booktitle={Proceedings of the 31st ACM SIGKDD Conference on Knowledge Discovery and Data Mining V. 2},
  pages={3483--3494},
  year={2025}
}

@article{christiano2017deep,
  title={Deep reinforcement learning from human preferences},
  author={Christiano, Paul F and Leike, Jan and Brown, Tom and Martic, Miljan and Legg, Shane and Amodei, Dario},
  journal={Advances in neural information processing systems},
  volume={30},
  year={2017}
}

@article{ouyang2022training,
  title={Training language models to follow instructions with human feedback},
  author={Ouyang, Long and Wu, Jeffrey and Jiang, Xu and Almeida, Diogo and Wainwright, Carroll and Mishkin, Pamela and Zhang, Chong and Agarwal, Sandhini and Slama, Katarina and Ray, Alex and others},
  journal={Advances in neural information processing systems},
  volume={35},
  pages={27730--27744},
  year={2022}
}

@inproceedings{azar2024general,
  title={A general theoretical paradigm to understand learning from human preferences},
  author={Azar, Mohammad Gheshlaghi and Guo, Zhaohan Daniel and Piot, Bilal and Munos, Remi and Rowland, Mark and Valko, Michal and Calandriello, Daniele},
  booktitle={International Conference on Artificial Intelligence and Statistics},
  pages={4447--4455},
  year={2024},
  organization={PMLR}
}

@article{ethayarajh2024kto,
  title={Kto: Model alignment as prospect theoretic optimization},
  author={Ethayarajh, Kawin and Xu, Winnie and Muennighoff, Niklas and Jurafsky, Dan and Kiela, Douwe},
  journal={arXiv preprint arXiv:2402.01306},
  year={2024}
}

@article{piratla2022focus,
  title={Focus on the Common Good: Group Distributional Robustness Follows},
  author={Piratla, Vihari and Netrapalli, Praneeth and Sarawagi, Sunita},
  journal={ICLR},
  year={2022},
  publisher={arXiv}
}

@article{loshchilov2017decoupled,
  title={Decoupled weight decay regularization},
  author={Loshchilov, Ilya and Hutter, Frank},
  journal={arXiv preprint arXiv:1711.05101},
  year={2017}
}

@article{jarvelin2002ndcg,
  title={Cumulated gain-based evaluation of IR techniques},
  author={J{\"a}rvelin, Kalervo and Kek{\"a}l{\"a}inen, Jaana},
  journal={ACM Transactions on Information Systems (TOIS)},
  volume={20},
  number={4},
  pages={422--446},
  year={2002},
  publisher={ACM}
}

\appendix

\section{Closed-form Update for Adaptive Worst-prefix Optimization}
\label{appendix: kkt}
\begin{proposition}[Closed-form update via KKT]
Consider the convex program
\[
\max_{w\in\Delta_T}\;\; F(w)
:= \sum_{m=1}^T w_m L_m - \frac1\eta \,KL(w\,\|\,w^t)
\]
where $\Delta_T=\{w\in\mathbb{R}^T_{\ge 0}:\sum_m w_m=1\}$ and assume $w^t\in\mathrm{int}(\Delta_T)$ (i.e., $w^t_m>0$ for all $m$). Then the unique maximizer is
\[
w^{t+1}_m
= \frac{w^t_m\exp\!\big(\eta L_m\big)}{\sum_{j=1}^T w^t_j\exp\!\big(\eta L_j\big)}.
\]
\end{proposition}

\begin{proof}
The objective is strictly concave on $\Delta_T$ because it is a sum of a linear term and the (strictly) concave entropy term
\(
-\frac1\eta\sum_m w_m\log w_m
\)
plus the linear term
\(
\frac1\eta\sum_m w_m \log w^t_m
\).
Hence the Karush--Kuhn--Tucker (KKT) conditions are necessary and sufficient.

We write the problem with constraints $\sum_m w_m=1$ and $w_m\ge 0$.
The Lagrangian is
\[
\mathcal{L}(w,\lambda,\nu)
= \sum_m w_m L_m - \frac1\eta\sum_m w_m\log\frac{w_m}{w_m^t}
+ \lambda\!\left(\sum_m w_m-1\right) - \sum_m \nu_m w_m,
\]
with multipliers $\lambda\in\mathbb{R}$ and $\nu_m\ge 0$.

\paragraph{Stationarity.}
For each $m$,
\[
\frac{\partial\mathcal{L}}{\partial w_m}
= L_m - \frac1\eta\!\left(\log\frac{w_m}{w_m^t}+1\right) + \lambda - \nu_m
= 0.
\]
Rearranging gives
\[
\log\frac{w_m}{w_m^t}= \eta\,(L_m+\lambda-\nu_m) - 1,
\quad\text{so}\quad
w_m = w_m^t \exp\!\big(\eta(L_m+\lambda-\nu_m) - 1\big).
\tag{$\ast$}
\]

\paragraph{Complementary slackness and interiority.}
% Because $w_m^t>0$ for all $m$ and the entropy regularization is strictly concave, the optimum lies in the interior of the simplex, i.e., $w_m>0$ for all $m$.
% Hence by complementary slackness $\,\nu_m w_m=0\,$ we must have $\nu_m=0$ for all $m$.
Because $w_m^t>0$ for all $m$, the gradient of the KL term
becomes singular as $w_m \to 0^+$.
Hence the maximizer cannot lie on the boundary of the simplex,
and the optimum satisfies $w_m>0$ for all $m$.

Plugging $\nu_m=0$ into $(\ast)$ yields
\[
w_m = \exp(\eta\lambda-1)\, w_m^t \exp(\eta L_m).
\]
The common factor $\exp(\eta\lambda-1)$ is fixed by the simplex constraint $\sum_m w_m=1$, giving
\[
\exp(\eta\lambda-1)
= \left(\sum_{j=1}^T w^t_j \exp(\eta L_j)\right)^{-1}.
\]
Therefore
\[
w^{t+1}_m
= \frac{w^t_m \exp(\eta L_m)}{\sum_{j=1}^T w^t_j \exp(\eta L_j)}.
\]

\paragraph{Feasibility and dual feasibility.}
By construction $w^{t+1}\in\Delta_T$ and $\nu_m=0\ge 0$ satisfy primal/dual feasibility. Since the objective is strictly concave, this solution is the unique maximizer.
\end{proof}

\section{Proof of Theorem~\ref{thm:consistency}}
\label{appendix_theorem1}

% \begin{theorem}[Optimization Consistency]
% Optimizing the prefix-aware loss (Eq.~(\ref{equ:overall_loss})) essentially optimizes a lower bound of the ranking metric (e.g., recall) under beam search. Therefore, optimizing the prefix-aware loss improves the ranking performance in the beam search setting.
% \end{theorem}

\begin{proof}
Recall the beam search success condition defined in Eq.~(\ref{eq:beam_recall}). A target item $y$ is successfully recalled if and only if its prefix ranks within the top-$K$ at every decoding step $t$. Let $A_m$ denote the event that the positive prefix survives the beam pruning at step $m$, formally defined as:
\[
A_m = \left\{ \operatorname{rank}_{m}(S_m(y^{<m})) \le K \right\}.
\]
Consequently, the indicator function for the beam search success of item $y$ is the product of these step-wise indicators:
\[
\mathbb{I}_{\text{Beam}}(y) = \prod_{m=1}^T \mathbb{I}(A_m).
\]
By applying the union bound inequality, the probability of the item being generated (i.e., the success condition) is lower-bounded by the probability that no failure occurs at any step:
\begin{equation}
\label{eq:union_lower_bound}
\mathbb{I}_{\text{Beam}}(y) 
= 1 - \mathbb{I}\left( \bigcup_{m=1}^T A_m^c \right) 
\;\ge\; 
1 - \sum_{m=1}^{T} \mathbb{I}(A_m^c),
\end{equation}
where $A_m^c$ represents the failure event at step $m$ (i.e., $\operatorname{rank}_m > K$).

Next, we establish the relationship between this discrete failure event and our continuous pairwise loss. Let $\mathcal{N}$ be the set of sampled negative prefixes and define $\phi_m$ as the log-sum-exp difference between the negative and positive prefix scores:
\[
\phi_m = \log \sum_{j \in \mathcal{N}} \exp\!\big( s_{j,-}^{m} - s_{i,+}^{m} \big).
\]
If the ranking failure $A_m^c$ occurs within the sampled set, it implies that the positive prefix score is surpassed by at least one negative prefix, i.e., $\max_j (s_{j,-}^{m} - s_{i,+}^{m}) > 0$. Since the log-sum-exp function upper bounds the maximum, we have $\phi_m \ge \max_j (s_{j,-}^{m} - s_{i,+}^{m})$. Therefore, the failure event implies a positive $\phi_m$:
\[
A_m^c \subseteq \{ \phi_m > 0 \} \iff A_m^c \subseteq \{ -\phi_m < 0 \}.
\]
This inclusion yields the inequality for the indicator functions:
\begin{equation}
\mathbb{I}(A_m^c) \;\le\; \mathbb{I}(-\phi_m < 0).
\end{equation}

To derive a differentiable optimization objective, we introduce the exponential surrogate function $\overline{\iota}(x)$, defined as:
\[
\overline{\iota}(x)
=
\exp\!\big(-(-x)_+\big)
=
\begin{cases}
1, & x \ge 0,\\
e^{x}, & x < 0.
\end{cases}
\]
The function $\overline{\iota}(x)$ strictly upper bounds the step function, i.e., $\overline{\iota}(x) \ge \mathbb{I}(x \ge 0)$. Applying this to our context (with $x = \phi_m$):
\[
\mathbb{I}(A_m^c) 
\;\le\; 
\mathbb{I}(-\phi_m < 0) 
\;\le\; 
\overline{\iota}(\phi_m) 
= 
\exp\!\big(-(-\phi_m)_+\big).
\]
Substituting this back into Eq.~(\ref{eq:union_lower_bound}), we obtain a sample-level lower bound for the beam search success:
\begin{equation}
\label{eq:surrogate_bound}
\mathbb{I}_{\text{Beam}}(y)
\;\ge\;
1 - \sum_{m=1}^{T}\exp\!\big(-(-\phi_m)_+\big).
\end{equation}

Finally, we link this bound to the prefix-aware pairwise loss defined in Eq.~(\ref{equ:pair_loss_def}). The loss can be rewritten as:
\[
\mathcal{L}_{\text{pair}}(m) = -\log \sigma(-\phi_m) = \log(1 + e^{\phi_m}).
\]
Let $\ell(z) = \log(1 + e^z)$. Since $e^{-(-z)_+} \le 1+e^z = e^{\ell(z)}$ holds for all real numbers $z$, we have:
\begin{equation}
\exp\!\big(-(-\phi_m)_+\big) \;\le\; \exp\!\big(\mathcal{L}_{\text{pair}}(m)\big).
\label{eq:pair_link}
\end{equation}
Combining Eq.~(\ref{eq:surrogate_bound}) and Eq.~(\ref{eq:pair_link}):
\[
\mathbb{I}_{\text{Beam}}(y) \;\ge\; 1 - \sum_{m=1}^T \exp\!\big(\mathcal{L}_{\text{pair}}(m)\big).
\]
This demonstrates that minimizing the prefix-aware pairwise loss $\mathcal{L}_{\text{pair}}(m)$ minimizes the upper bound of the failure probability, thereby explicitly maximizing the lower bound of the beam search recall $\mathbb{I}_{\text{Beam}}(y)$.

For the pointwise variant, the proof follows an analogous logic by viewing the pointwise cross-entropy as a pairwise loss defined against the entire vocabulary (where all non-target tokens serve as negatives).
\end{proof}

\section{Additional Implementation Details}
\label{appendix:implement}

\noindent \textbf{Model Architectures.}
For the TIGER backbone, we adopt a Transformer-based encoder-decoder architecture. Both the encoder and decoder consist of 4 layers, each featuring 6 self-attention heads with a head dimension of 64. The ReLU activation function is applied across all layers. For the Llama backbone, we utilize an 8-layer Transformer decoder. The learnable token embedding dimension is set to 128 for all backbones.

\noindent \textbf{Baseline Settings.}
Regarding the baselines, for MSL, the temperature parameter $\tau$ is tuned over the set $\{1.0, 2.0, 3.0\}$. For DPO, DMPO, and S-DPO, we perform supervised fine-tuning (SFT) with cross-entropy loss as a prerequisite, following the protocol in~\cite{chen2024SDPO}. The hyperparameter $\beta$ is selected from $\{0.5, 1.0, 2.0\}$.

The optimal hyperparameters of \ourmethod{} for each dataset are reported in~\autoref{tab:reproduction}.

\begin{table}[h]
\centering
\small
\caption{Best hyperparameters of \ourmethod{} for each dataset.}
\label{tab:reproduction}
\renewcommand\arraystretch{1.1}
\tabcolsep=0.06cm
\resizebox{0.48\textwidth}{!}{
\begin{tabular}{llcccccccc}
\toprule
\multirow{2}{*}{\textbf{Backbone}} & \multirow{2}{*}{\textbf{Methods}} & \multicolumn{2}{c}{\textbf{Office}} & \multicolumn{2}{c}{\textbf{Grocery}} & \multicolumn{2}{c}{\textbf{Beauty}} & \multicolumn{2}{c}{\textbf{Yelp}} \\
\cmidrule(r){3-4}\cmidrule(r){5-6}\cmidrule(r){7-8}\cmidrule(r){9-10}
& & $\beta$ & $\eta$ & $\beta$ & $\eta$ & $\beta$ & $\eta$ & $\beta$ & $\eta$ \\
\midrule
\multirow{2}{*}{TIGER} & \ourmethod{}-Pointwise  & 0.3 & $1e^{-4}$ & 0.1 & $5e^{-5}$ & 0.1 & $1e^{-5}$ & 0.05 & $1e^{-5}$ \\
\cmidrule(r){2-10}
 & \ourmethod{}-Pairwise  & 0.2 & $3e^{-5}$ & 0.1 & $5e^{-5}$ & 0.1 & $1e^{-4}$ & 0.05 & $5e^{-6}$ \\
\midrule
\multirow{2}{*}{Llama} & \ourmethod{}-Pointwise  & 0.4 & $5e^{-4}$ & 0.4 & $5e^{-5}$ & 0.3 & $1e^{-4}$ & 0.4 & $1e^{-5}$ \\
\cmidrule(r){2-10}
& \ourmethod{}-Pairwise  & 0.4 & $5e^{-5}$ & 0.2 & $3e^{-5}$ & 0.4 & $1e^{-5}$ & 0.1 & $5e^{-6}$ \\
\bottomrule
\end{tabular}
}
\end{table}

\section{Additional Discussion.}
\label{appendix: discussion}
\noindent \textbf{Relation to preference optimization baselines.}
While \ourmethod{}-pairwise shares structural similarities with S-DPO~\cite{chen2024SDPO}—specifically the use of a softmax-based ranking loss over multiple negatives—they differ fundamentally in their \textbf{optimization granularity}.

S-DPO optimizes preferences at the \textbf{complete sequence level} (i.e., the final item). It computes the loss based on the cumulative likelihood of complete item titles, involving only a single summation over the negative candidate set $\mathcal{N}$:
\begin{equation}
\begin{aligned}
    \label{equ:s_dpo}
    & \mathcal{L}_{\text{S-DPO}} = \\
    & -\log \sigma \left( - \log \underbrace{\sum_{j \in \mathcal{N}} \exp \left( \beta \log \frac{\pi_\theta(e_{j, -}|x_u)}{\pi_{\text{ref}}(e_{j, -}|x_u)} - \beta \log \frac{\pi_\theta(e_{i, +}|x_u)}{\pi_{\text{ref}}(e_{i, +}|x_u)} \right)}_{\text{Sum over negative items (Sequence-level)}} \right),
\end{aligned}
\end{equation}
where $e_{i, +}$ and $e_{j, -}$ represent the complete preferred and dispreferred item sequences, respectively.

In contrast, our method performs optimization at \textbf{every intermediate prefix step}. This dense supervision is essential for autoregressive generation tasks like beam search, where early deviations can lead to error propagation. By aligning preferences at each token position, \ourmethod{} ensures the model stays on the correct decoding path. The unified objective is defined as:

\begin{equation}
\begin{aligned}
\label{equ:appendix_final_pairwise}
\mathcal{L}_\text{pair}
= \mathcal{L}_{\mathrm{CE}}
+ \beta \cdot
\underbrace{\sum_{m=1}^{T} w_m}_{%
\smash{\raisebox{-0.1ex}{$\scriptstyle\text{\clap{Sum over prefixes}}$}}%
}
\Bigl[
-\log \sigma \Bigl(
-\log
\underbrace{
\sum_{j \in \mathcal{N}_m}
\exp\left(s_{j,-}^m - s_{i,+}^m \right)
}_{%
\smash{\raisebox{-0.1ex}{$\scriptstyle\text{\clap{Sum over negative prefixes (Step-level)}}$}}%
}
\Bigr)
\Bigr].
\end{aligned}
\end{equation}
where $m$ indexes the prefix steps and $s^m$ represents the score at the specific step $m$.

\section{Learning Algorithm of \ourmethod{}}
\label{appendix:learning_algo}

Algorithm~\ref{alg:apao_unified} presents the unified learning procedure of \ourmethod{}, which accommodates both \textbf{Pointwise} and \textbf{Pairwise} variants through a \texttt{Mode} parameter. To facilitate implementation, we integrate PyTorch-style pseudocode specifically for the core loss calculations and the adaptive weighting mechanism.

\begin{algorithm}[h]
\scriptsize
\caption{Learning Algorithm of \ourmethod{}}
\label{alg:apao_unified}
\begin{algorithmic}[1]
\Require Training data $\mathcal{D}$; generator parameters $\theta$;
loss weight $\beta$; update rate $\eta$; mode $M$
\Statex \hspace{\algorithmicindent}
where $M\in\{\texttt{Pointwise},\texttt{Pairwise}\}$
\Ensure Optimized parameters $\theta$
\State Initialize $\theta$ and prefix weights
$\mathbf{w}\leftarrow[1/T,\ldots,1/T]$

\For{each mini-batch $(x,y)\in\mathcal{D}$}
    \If{$M=\texttt{Pairwise}$}
        \State Sample negatives $\{y_k^-\}_{k=1}^{K}$ for each $y$
        \State $\mathcal{C}\leftarrow\{y\}\cup\{y_k^-\}_{k=1}^{K}$
    \Else
        \State $\mathcal{C}\leftarrow\{y\}$
    \EndIf
    \State $\texttt{logits}\leftarrow \mathrm{Model}_{\theta}(x,\mathcal{C})$
    \State $\mathcal{L}_{\mathrm{CE}}\leftarrow
    \mathrm{CrossEntropy}(\texttt{logits of }y,y)$
    \State $\mathbf{L}_{\mathrm{prefix}}\leftarrow
    \textproc{CalcPrefixLoss}(\texttt{logits},\mathcal{C},M)$
    \State $\mathbf{w}\leftarrow
    \textproc{AdaptiveUpdate}(\mathbf{w},\mathbf{L}_{\mathrm{prefix}},\eta)$
    \State $\mathcal{L}_{\mathrm{total}}\leftarrow
    \mathcal{L}_{\mathrm{CE}}+
    \beta\cdot(\mathbf{w}\cdot\mathbf{L}_{\mathrm{prefix}}).\mathrm{sum}()$
    \State Update $\theta$ w.r.t. $\mathcal{L}_{\mathrm{total}}$
\EndFor

\Statex
\Function{CalcPrefixLoss}{\texttt{logits, cands, mode}}
    \State \texttt{log\_probs = F.log\_softmax(logits, dim=-1)}
    \State \texttt{token\_logp = log\_probs.gather(2, cands.unsqueeze(-1)).squeeze(-1)}
    \State \texttt{losses = []}
    
    \For{\texttt{t in 1 to T}}
        \If{\texttt{mode == 'Pointwise'}}
            \State \texttt{prefix\_logp = token\_logp[:, :t].sum(dim=1) / t}
            \State \texttt{loss\_t = -prefix\_logp.mean()}
        \Else \hfill \textcolor[RGB]{75,123,128}{\# Pairwise Mode}
            \State \texttt{prefix\_score = token\_logp[:, :t].sum(dim=1)}
            \State \texttt{pos\_score = prefix\_score[0]}
            \State \texttt{neg\_scores = prefix\_score[1:]}
            \State \texttt{pos\_exp = torch.exp(pos\_score)}
            \State \texttt{neg\_sum\_exp = torch.exp(neg\_scores).sum()}
            \State \texttt{denom = pos\_exp + neg\_sum\_exp + 1e-12}
            \State \texttt{loss\_t = -torch.log(pos\_exp / denom)}
            \State \texttt{loss\_t = loss\_t.mean()}
        \EndIf
        \State \texttt{losses.append(loss\_t)}
    \EndFor
    \State \Return \texttt{torch.stack(losses)}
\EndFunction

\Statex
\Function{AdaptiveUpdate}{\texttt{w\_last, losses, eta}}
    \State \texttt{base\_loss = losses.detach()}
    \State \texttt{w\_new = w\_last * torch.exp(eta * base\_loss)}
    \State \texttt{w\_new = w\_new / w\_new.sum()}
    \State \Return \texttt{w\_new}
\EndFunction

\end{algorithmic}
\end{algorithm}

\end{document}